\documentclass[11pt]{article}

\usepackage[margin=0.82in]{geometry}
\usepackage{thmtools,thm-restate}

\usepackage{amssymb,amsthm,amsmath,amssymb,wrapfig,dsfont,authblk}
\usepackage{mathtools}
\DeclarePairedDelimiter{\ceil}{\lceil}{\rceil}
\usepackage{thmtools,thm-restate}
\PassOptionsToPackage{hyphens}{url}
\usepackage[square,sort,semicolon,numbers]{natbib}
\usepackage{varioref}
\usepackage[usenames,svgnames,xcdraw,table]{xcolor}
\definecolor{DarkBlue}{rgb}{0.1,0.1,0.5}
\definecolor{DarkGreen}{rgb}{0.1,0.5,0.1}
\usepackage[backref=page]{hyperref}
\hypersetup{
     colorlinks   = true,
     linkcolor    = DarkBlue, 
     urlcolor     = DarkBlue, 
	 citecolor    = DarkGreen 
}
\renewcommand*{\backref}[1]{}
\renewcommand*{\backrefalt}[4]{%
    \ifcase #1 (Not cited.)%
    \or        (Cited on page~#2)%
    \else      (Cited on pages~#2)%
    \fi}
\usepackage[capitalise,noabbrev]{cleveref}

\usepackage{enumerate}
\usepackage{subcaption}
\usepackage{mathtools}
\usepackage{breqn}

\usepackage{bbm}
\usepackage{bbold}

\usepackage{tabulary}
\usepackage{multirow}
\usepackage{booktabs}
\usepackage{colortbl}

\usepackage{algorithmic}
\usepackage{algorithm}

\usepackage{verbatim}

\newtheorem{lemma}{Lemma}
\newtheorem{claim}{Claim}

\newtheorem{definition}{Definition}
\theoremstyle{definition}

\allowdisplaybreaks

\newcommand{\A}{\mathcal A}

\newcommand{\G}{{\mathcal G}}

\newcommand{\I}{{\mathcal I}}

\newcommand{\R}{\mathcal R}

\newcommand{\T}{\mathcal T}

\newcommand{\W}{\mathcal{W}}
\newcommand{\E}{\mathcal{E}}
\newcommand{\mL}{\mathcal{L}}

\newcommand{\ALG}{\textsc{Alg}}

\newcommand{\EF}{{\textrm{EF}}}

\newcommand{\linear}{\textsc{Ef-Linear}}
\newcommand{\PPAD}{{\rm PPAD} }
\newcommand{\PLS}{{\rm PLS} }

\newcommand{\plc}{\textsc{Ef-PiecewiseLinear}}

\newcommand{\tup}{(\pi, L,U)}

\begin{document}
	\title{{\bfseries Fully Polynomial-Time Approximation Schemes for \\ Fair Rent Division}}
	\author{Eshwar Ram Arunachaleswaran\thanks{Indian Institute of Science. {\tt eshwarram.arunachaleswaran@gmail.com}} \qquad Siddharth Barman\thanks{Indian Institute of Science. {\tt barman@iisc.ac.in}} \qquad Nidhi Rathi\thanks{Indian Institute of Science. {\tt nidhirathi@iisc.ac.in}}}
	\date{}
	\maketitle

\begin{abstract}
We study the problem of fair rent division that entails splitting the rent and allocating the rooms of an apartment among roommates (agents) in a \emph{fair} manner. In this setup, a distribution of the rent and an accompanying allocation is said to be fair if it is envy free, i.e., under the imposed rents, no agent has a strictly stronger preference for any other agent's room. The cardinal preferences of the agents are expressed via functions which specify the utilities of the agents for the rooms for every possible room rent/price. While envy-free solutions are guaranteed to exist under reasonably general utility functions, efficient algorithms for finding them were known only for \emph{quasilinear utilities}. This work addresses this notable gap and develops approximation algorithms for fair rent division with minimal assumptions on the utility functions. 

Specifically, we show that if the agents have continuous, monotone decreasing, and piecewise-linear utilities, then the fair rent-division problem admits a fully polynomial-time approximation scheme (FPTAS). That is, we develop algorithms that find allocations and prices of the rooms such that for each agent $a$ the utility of the room assigned to it is within a factor of $(1 + \varepsilon)$ of the utility of the room most preferred by $a$. Here, $\varepsilon>0$ is an approximation parameter, and the running time of the algorithms is polynomial in $1/\varepsilon$ and the input size. In addition, we show that the methods developed in this work provide efficient, truthful mechanisms for special cases of the rent-division problem. Envy-free solutions correspond to equilibria of a two-sided matching market with monetary transfers; hence, this work also provides efficient algorithms for finding approximate equilibria in such markets. We complement the algorithmic results by proving that the fair rent division problem (under continuous, monotone decreasing, and piecewise-linear utilities) lies in the intersection of the complexity classes \PPAD and \PLS. 

\end{abstract}

\section{Introduction}
\label{section:intro}

Fair division addresses the fundamental problem of allocating goods among agents with equal entitlements, but distinct preferences. 
Such resource-allocation settings have been studied over multiple decades in economics, mathematics, and computer science; see, e.g.,~\cite{brams1996fair}, \cite{brandt2016handbook} and \cite{moulin2016handbook}. In this line of work, a classic problem---with direct practical implications\footnote{For example, the website Spliddit, \url{http://www.spliddit.org/}, has been used over thirty thousand times for fair rent division (as of June 2018) \cite{Gal2017}.}---is to allocate rooms (indivisible goods) among agents while also assigning the rent (setting prices) in a \emph{fair} manner.  


The standard notion of fairness in this setting is \emph{envy-freeness} (introduced in \cite{foley1967resource} and studied in \cite{varian1974equity} and \cite{stromquist1980cut}) which requires that, under the imposed rents, each agent prefers the room allocated to it over that of any other agent.  The preferences of the agents are expressed via functions (one for every agent-room pair) which specify the utilities of the agents for the rooms at every possible room rent/price. Hence, for the rent division problem, we say that an assignment of the rooms and the rent is envy free (i.e., fair), if the utility that each agent $a$ derives from the room allocated to it (under the imposed rents) is at least as high as $a$'s utility for any other room. 

Prior work has established that envy-free rent divisions are guaranteed to exist under general, but well-behaved, utility functions. In particular, Sun and Yang \cite{sun2001fair} have shown that if all the utility functions in a rent-division instance are continuous, monotone decreasing, and bounded, then an envy-free solution will necessarily exist. This result relies on a fixed-point argument--namely, the Knaster, Kuratowski, and Mazurkiewicz (KKM) Lemma \cite{knaster1929beweis}. Similar existence guarantees have been established by Svensson \cite{svensson1983large} and Alkan et al. \cite{alkan1991fair}.\footnote{A key tool in the work of Alkan et al. \cite{alkan1991fair} is a perturbation lemma (\cref{lemma:original-perturb}), which is utilized in this work as well.} 
 
The special case of \emph{quasilinear utilities} has also received attention in the rent-division context; see, e.g., \cite{aragones1995derivation}, \cite{abdulkadirouglu2004room}, \cite{Gal2017}, and~\cite{procacciafair}. In this utility model each agent $a$ has a base value for every room $r$, and $a$'s utility for $r$ at price $p_r$ is equal to the base value minus $p_r$. For quasilinear utilities, the work of Aragones \cite{aragones1995derivation} provides a combinatorial proof of existence along with an efficient algorithm for finding envy-free solutions. 

It is relevant to note that while the existence result holds under fairly general utilities, efficient algorithms were known only for the quasilinear case. This work addresses this notable gap and develops approximation algorithms for fair rent division with essentially minimal assumptions on the utility functions. In particular, our results hold as long as the utility functions are continuous, monotone decreasing, and piecewise linear. Note that, for discontinuous and unbounded functions the existence of an envy-free solution cannot be guaranteed (see \cref{section:counter} for examples). Also, the assumption of piecewise linearity ensures that the underlying utility functions can be explicitly provided as input. 

Indeed, the utility functions considered in this work  are not confined to be concave (or convex), and can be used to heterogeneously expresses agents' preferences at different price ranges, e.g., we can model agents who have quasilinear utilities and a fixed budget by considering piecewise-linear functions that experience a sharp drop when the price reaches the budget; Procaccia et al.~\cite{procacciafair} provide a specialized algorithm to address these budget constraints. By contrast, this setting can be modeled as a special case of the utilities considered in this work.  

The rent division problem can be stated abstractly in terms of dividing indivisible goods (the rooms) along with money (the rents) among unit-demand agents (i.e., each agent wishes to acquire at most one item). This perspective is adapted in the work of Svensson \cite{svensson1983large}, Alkan et. al. \cite{alkan1991fair}, Sun and Yang \cite{sun2001fair}, and Aragones \cite{aragones1995derivation}, who also consider relevant  variants of the rent-division problem, such as characterizing \emph{optimal prices} and \emph{truthful mechanisms}. This paper addresses these variants and, in particular, focuses on (i) finding an envy-free solution with nonnegative prices and nonnegative utilities along with (ii) finding an envy-free solution which splits a given total rent.

This work shows that both of these problems admit a fully polynomial-time approximation scheme (FPTAS). Formally, the developed algorithms find allocations of the rooms and rents such that for each agent $a$ the utility of the room assigned to it is within a factor of $(1+\varepsilon)$ of the utility of the room most preferred by $a$. Here, $\varepsilon >0$ is an approximation parameter and the running time of the developed algorithms is polynomial in $1/\varepsilon$ and the input size. Overall, we show that a natural, approximate analogue of envy freeness can be achieved efficiently for the broad class of utilities mentioned above. The following list summarizes our contributions. \\


\noindent 
{\bf Our Results and Techniques}
Throughout, we study the problem of fair rent division under continuous, monotone decreasing, and piecewise-linear utility functions. We will also conform to the standard assumption that every agent's utility for each room is nonnegative when the room rent is zero. 
\begin{itemize}
\item We design an algorithm for finding an approximately envy-free solution in which the rents and the utilities of the agents are nonnegative (\cref{theorem:main-result}). This approximation guarantee asserts that for each agent $a$ the utility of the room assigned to it is within a factor of $(1+\varepsilon)$ of the utility of the room most preferred by $a$. The runtime of our approximation algorithm is polynomial in $1/\varepsilon$ and the input size. 

\item We also develop an FPTAS for the problem of computing an envy-free solution in which the sum of prices (room rents) is equal to a specified total rent (\cref{theorem:fixed-cost}). In contrast to the previous result, the solution obtained for this problem might impose negative prices, i.e., require \emph{transfers} among agents. Note that there exist rent-division instances in which every envy-free solution, which splits the total rent, levies negative rents on some room.

Under this total-rent constraint, nonnegativity of utilities under fair solutions cannot be guaranteed either. To circumvent this issue a natural scaling assumption is adopted in prior work on quasilinear utilities, see, e.g., \cite{brams2001competitive}. We show that, even in our setting, if the given rent-division instance satisfies this assumption, then we can find an approximately envy-free solution which not only splits the given total rent, but also yields nonnegative utilities (\cref{theorem:ir}).

The two results mentioned above are obtained by first considering \emph{structured instances} in which the constituent slopes of all the piecewise-linear utilities are integer powers of $(1+\varepsilon)$. We develop an algorithm that finds (exact) envy-free solutions of such instances in time that is polynomial in $1/\varepsilon$ and the input size (\cref{theorem:ef-rounded}). We then show that any given rent-division instance $\mathcal{I}$ can be \emph{rounded} to obtain a structured instance $\overline{\mathcal{I}}$ such that an envy-free solution of $\overline{\mathcal{I}}$ is an approximately envy-free solution of $\mathcal{I}$ (\cref{lemma:apx-guarantee}).

\item Sun and Yang~\cite{sun2003general} proved a somewhat surprising result that every rent-division instance admits an \emph{optimal solution} which is  simultaneously envy free, \emph{efficient}, and \emph{nonmanipulable}. The optimal solution is efficient in the sense that the utility profile it induces weakly Pareto dominates the utility profile of any other envy-free solution with nonnegative rents; the optimal solution itself imposes nonnegative rents. Here, nonmanipulability refers to the property that any algorithm that select the optimal solution as its outcome is guaranteed to be dominant strategy incentive compatible (DSIC). We show that for structured instances optimal prices can be computed efficiently. Therefore, for such instances we obtain a DSIC mechanism for the fair rent-division problem.  Our result is based on a novel characterization of optimal solutions, which not only provides an alternate proof of existence of such solutions, but also leads to an efficient algorithm.\footnote{The result of Green and Laffont~\cite{green1979incentives} rules out universal existence of DSIC mechanisms that are both Pareto efficient and budget balanced (i.e., split a given total rent). Hence, in the DSIC result of Sun and Yang~\cite{sun2003general}---and its algorithmic version obtained in this work---the sum of rents cannot be fixed a priori.}

\item The developed algorithm also provides an efficient method to find (exact) envy-free solutions when the number of agents is fixed (\cref{subsection:special}). Note that in this special case the piecewise-linear utilities can still be intricate. We additionally show that if in a given rent-division instance the number of distinct slopes (across utility functions) is a fixed constant, then an envy-free solution can be computed in polynomial time (\cref{subsection:special}).  
 
\item We complement the algorithmic results by proving that the \emph{total problem} of finding an envy-free solution of a rent-division instance (with continuous, monotone decreasing, and piecewise-linear utilities) is contained in the complexity class \PPAD (Polynomial Parity Arguments on Directed graphs) as well as \PLS (Polynomial Local Search); see \cref{theorem:complexity}. Though, prior work has established the existence of fair solutions using a fixed-point argument---in particular, the KKM Lemma~\cite{sun2001fair}---the containment of the (exact) fair rent division in \PPAD is not an immediate consequence of this proof. This follows the fact that the computational version of the KKM lemma would entail discretization of the solution space (as is required in the case of related problems such as Sperner's Lemma~\cite{papadimitriou1994complexity} and Envy-Free Cake Cutting~\cite{deng2012algorithmic}) and, hence, we would get an approximate solution of the reduced problem. Therefore, in and of itself, such a reduction would imply that approximating rent division is in {\rm PPAD}---this result would not rule out the possibility that the exact version is harder. We bypass this issue by establishing a reduction from fair rent division to the problem of computing an exact Nash equilibrium in polymatrix games. Given that finding a Nash equilibrium is such games is known to be {\rm PPAD}-complete~\cite{daskalakis2006game}, we get the desired containment. This reduction might be of independent interest, since it provides an alternate proof of existence of envy-free solutions using Nash's theorem (specifically for polymatrix games), rather than a more involved fixed-point argument. 

Another conceptually interesting part of the complexity analysis is the containment of fair rent division in the complexity class {\rm PLS}. This containment relies on a potential argument which is developed in this paper. Overall, these results render fair rent division as one of the few ``natural'' problems that lie in $\PPAD \cap \PLS$ and for which a polynomial-time algorithm is not known.\footnote{Another problem with the same complexity status is \emph{Colorful Carath\'{e}odory}~\cite{meunier2017rainbow}. Showing that these problems are complete for a semantic subclass of \PPAD $\cap$ \PLS remains an interesting open question.}
\end{itemize}

As mentioned previously, we develop an exact algorithm for finding envy-free solutions for problem instances in which the slopes of all the piecewise-linear utilities are integer powers of $(1+\varepsilon)$. We also show that general rent-division instances can be rounded to such structured instances while approximately maintaining  envy freeness of solutions. This, overall, provides an FPTAS for fair rent division. 


Our exact algorithm for finding an envy-free solution (\cref{algorithm:envy-free}) begins by employing a homotopy idea. Specifically, we pick a large  threshold $M$ and for each utility function $v$ construct a surrogate function $\widehat{v}$ such that $v$ and $\widehat{v}$ have the same function value for all prices less than $M$. For prices greater than $M$, we set $\widehat{v}$ to be quasilinear. Recall that an envy-free solution with a nonnegative price vector can be computed efficiently if the utilities are quasilinear~\cite{aragones1995derivation}. Therefore, we can find a nonnegative envy-free price vector $p^0$ for the quasilinear ends of $\widehat{v}$s. Since a uniform, additive shift in prices maintains envy freeness under quasilinear utilities, we can add $M$ to each component of $p^0$ and ensure that it provides an initial envy-free solution for the utilities $\widehat{v}$s. 




The key part of the algorithm is to iteratively reduce the prices while maintaining envy freeness under $\widehat{v}$s. We continue to perform such a price reduction till all the rents are less than $M$. At this point, the solution is envy free with respect to both the constructed utilities $\widehat{v}$s and the given utilities $v$s.  Alkan et al. \cite{alkan1991fair} used a perturbation lemma (\cref{lemma:original-perturb}) to accomplish an arbitrarily small reduction in prices. By contrast, our algorithm performs this update by solving linear programs.\footnote{In and of itself, the perturbation lemma does not lead to a finite-time algorithm for finding fair solutions; Alkan \cite{alkan1989existence} provides an instance wherein the perturbations do not converge. Note that this lemma is not directly instantiated in our algorithms, however it is used to prove the correctness of the developed methods.}


Our results rely on establishing interesting geometric properties of the set of prices that induce envy-free solutions. Though this set is nonconvex, it consists   of a sequence of polytopes which successively intersect; see \cref{figure:non-convex}. These intersections necessarily include the price vectors at which the sum of prices is minimized in each polytope. At a high level, our algorithm obtains the price reduction by performing a ``walk'' over these polytopes. Each step of the walk entails minimizing the sum of prices over the current polytope (to reach an intersection point) and, subsequently, switching to a new polytope. The minimization step corresponds to solving a linear program and the switching step is implemented by solving a maximum-weight perfect matching problem (see \cref{algorithm:envy-free} for details).  

\begin{figure}[h]
	\begin{center}
		\includegraphics[scale=.5]{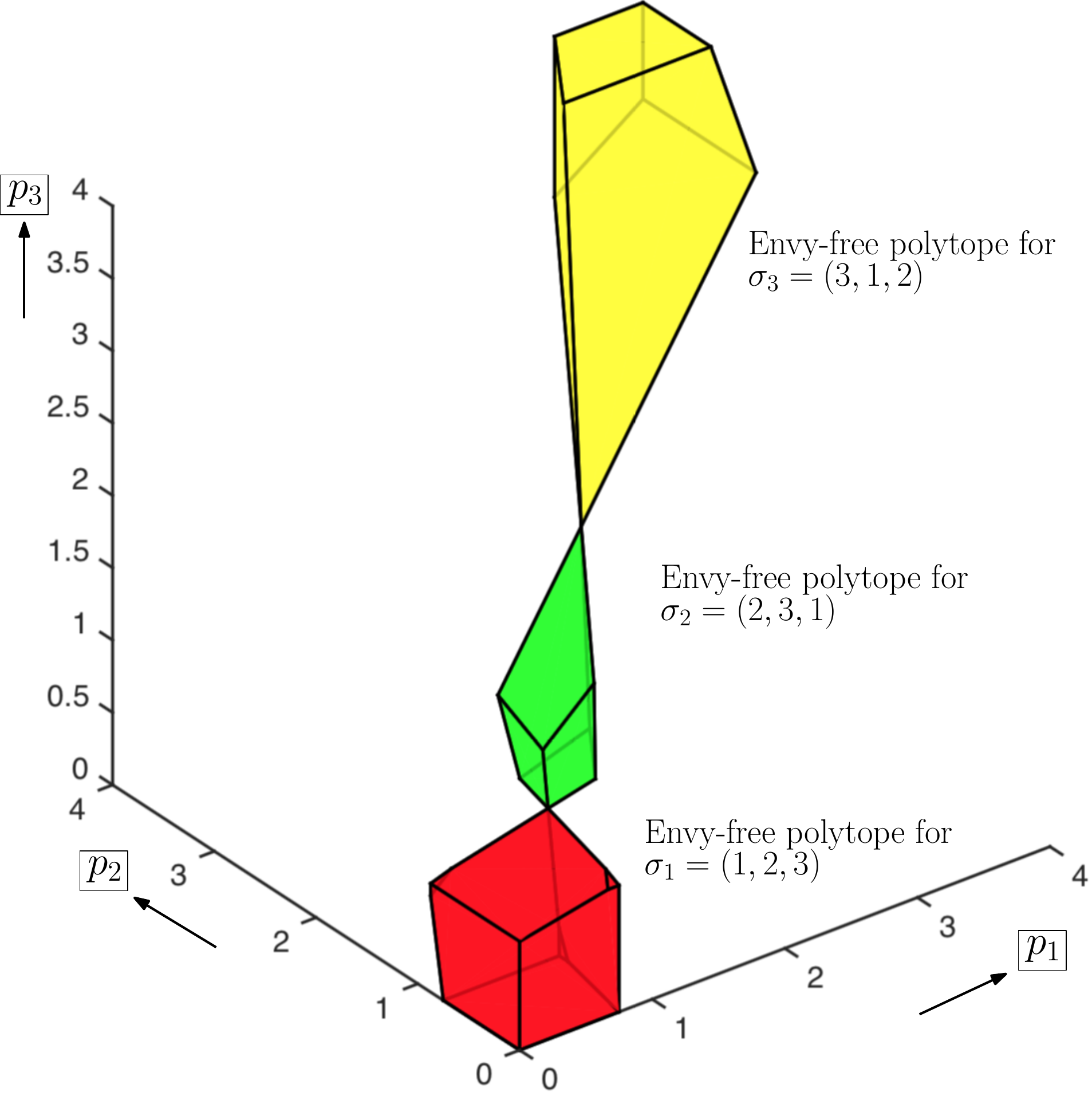}
	\end{center}
\caption{The figure depicts (a) nonconvexity of the set of prices that induce envy-free solutions, (b) chaining of polytopes, and (c) intersection of the polytopes contains the vertex that minimizes the sum of prices. Numerical details of this example appear in Appendix~\ref{section: Non-convexity}.}
\label{figure:non-convex}
\end{figure}

Another contribution of this work is to show that there exists a potential which decreases as the walk though the polytopes progresses. This potential argument not only bounds the runtime of the developed algorithm, but also enables us to show that the rent-division problem is contained in {\rm PLS}. 

It is relevant to note that the resulting algorithm is rather simple, even though the analysis is somewhat intricate---the algorithm primarily involves finding maximum-weight perfect matchings and solving linear programs. As mentioned previously, rent-division algorithms have been widely used in practice. A notable example is the popular website Spliddit,\footnote{\url{http://www.spliddit.org}. See also \url{https://www.nytimes.com/interactive/2014/science/rent-division-calculator.html}, \url{http://acritch.com/rent}, and \url{https://www.splitwise.com/calculators/rent}.} the scale of which highlights that efficient methods for fair rent division---like the ones developed in this work---have a potential for direct impact. \\


\noindent 
{\bf Further Related Work}
The work of Su \cite{su1999rental} establishes the existence of envy-free solutions in an ordinal version of rent-division problem. In this setup, each agent has a preference (over the rooms) for every possible division of the total rent. This existence result requires that the preferences constitute a closed set and satisfy the assumption that each agent is \emph{miserly}, i.e., prefers a free room (a room with zero rent) to a non-free room. Note that in this setting the preferences can be nonmonotone and might not admit a succinct, explicit representation. 

This ordinal setup is incomparable with the cardinal utility model considered in this work; in particular, the miserly assumption does not hold for arbitrary  quasilinear utilities. Furthermore, the ordinal version of the rent-division problem is {\rm PPAD}-complete.\footnote{A reduction along the lines of the one given in Deng et al. \cite{deng2012algorithmic} establishes this hardness result.} By comparison, we show that cardinal version considered in this paper is contained in the complexity class $\PPAD \cap \PLS$. \\

\noindent 
{\bf Competitive Equilibria of Two-Sided Matching Markets:} Fair rent division is mathematically equivalent to a two-sided matching market with monetary transfers. In this market formulation each agent is assumed to be unit demand and each agent's utility is a function of the indivisible good she receives as well as the amount of money she pays. Quinzii~\cite{quinzii1984core} established the existence of competitive equilibria under general (non-quasilinear) utilities in this market framework.  Furthermore, Demange and Gale~\cite{demange1985strategy} showed that equilibria in these markets admit a lattice structure.  A constructive proof of existence of such equilibria was established by Alkan~\cite{alkan1989existence}; in particular, Alkan's algorithm~\cite{alkan1989existence} finds an equilibrium in finite time, though it does not admit a polynomial running time bound. The quasilinear version of such two-sided matching markets has also been studied in prior work~\cite{gale1960theory, shapley1971assignment, demange1982strategyproofness}.

An envy-free solution corresponds to an equilibria of such a market, since envy freeness implies that each agent is maximizing its utility under the imposed prices. Therefore, this work also provides novel results for efficiently finding approximate equilibria of such two-sided matching markets under a broad class of utility functions. 


\section{Notation and Preliminaries}
\label{section:notation}

\noindent
{\bf Problem Instance} A \emph{rent-division instance} is a tuple $\langle \mathcal{A}, \mathcal{R}, \{ v_a(r, \cdot) \}_{a\in \mathcal{A}, r \in \mathcal{R}} \  \rangle$ wherein $\mathcal{A} =[n]$  denotes the set of $n$ agents and $\mathcal{R} =[n]$ denotes the set of $n$ rooms. The cardinal preference of each agent $a \in \mathcal{A}$ for every room $r \in  \mathcal{R}$ is specified via a utility function $v_a(r, \cdot)$; specifically, agent $a$'s utility for room $r$ at price (rent) $p_r \in \mathbb{R}$ is $v_a(r, p_r) \in \mathbb{R}$. 

Throughout, we will assume that the utility functions are continuous, monotone decreasing, and piecewise linear. Each function $v_a(r, \cdot) $ is given as an input via its constituent linear pieces. In particular,  $v_a(r, \cdot)$ is specified using its base value $v_a(r,0)$, a set of increasing \emph{break points}, $b_1 = 0, b_2, b_3, \ldots, b_t \in \mathbb{R}_+$, and the magnitude of the slopes $\{ \lambda^a_{r,i} \in \mathbb{R}_+ \}_{i \in [t]}$. Here, the utility in the interval $[b_{i}, b_{i+1}]$ is a linear function with slope $-\lambda^a_{r,i}$. Therefore, if price $x \in \mathbb{R}_+$ is contained in, say, the $i$th interval, $x \in [b_i, b_{i+1}]$, then the agent's utility for the room at $x$ is  
\begin{align*}
v_a(r, x) & \coloneqq v_a(r,0) - \sum_{j=1}^{i-1} \lambda^a_{r,j}  (b_{j+1} - b_j) \ - \lambda^a_{r,i} (x - b_i) 
\end{align*}

For ease of presentation, we will drop the dependency on $i$ and use $\lambda^a_r$ to denote $\lambda^a_{r,i}$ when the interval $i$ is clear from context.

Rent division has been extensively studied for the special case of {quasilinear utilities}. Here, every agent $a$ has a base value $v_a(r,0)$ for each room $r$, and the utility functions are quasilinear in the prices, $v_a(r, p_r) \coloneqq v_a(r,0) - p_r$. An efficient algorithm for quasilinear utilities  was developed in~\cite{aragones1995derivation}. Our main algorithm solves the quasilinear case as a subroutine; for completeness, we outline a polynomial-time algorithm for fair rent division under quasilinear utilities in \cref{section:quasi}. \\

\noindent
{\bf Allocations and Envy-Free Solutions} An \emph{allocation} refers to a bijection $\pi: \mathcal{A} \mapsto \mathcal{R}$ where room $\pi(a) \in \mathcal{R}$ is assigned to agent $a \in \mathcal{A}$. Furthermore, a \emph{solution}, $(\pi, p)$, to a rent-division instance is an allocation $\pi$ along with a price vector $p=(p_1, p_2, \ldots, p_n)$ for the $n$ rooms. The utility that agent $a$ achieves under the solution $(\pi, p)$ is $v_a(\pi(a), p_{\pi(a)})$. We will primarily consider solutions wherein the rents are nonnegative, $p_r \geq 0$ for all $r$; exceptions to this convention will be explicitly mentioned. 

Recall that {envy freeness} requires that each agent prefers its own ``share'' over that of any other agent. Formally, for a rent division instance $\langle \mathcal{A}, \mathcal{R}, \{ v_a(r, \cdot) \}_{a, r} \rangle$, a solution $(\pi, p)$ is said to be \emph{envy free} ($\EF$) iff, under the imposed rents, every agent prefers the room assigned to it over any other room, i.e., for each $a \in \mathcal{A}$ and every room $r\in R$ we have $v_a(\pi(a), p_{\pi(a)}) \geq v_a(r, p_r)$.

A direct characterization of envy-free solutions is obtained by considering a bipartite graph $\mathcal{F}(p) :=(\mathcal{A} \cup \mathcal{R}, F)$ where edge $(a,r)$ is included in the edge set $F$ iff $r$ is a maximum utility room for agent $a$, i.e., $(a,r) \in F$ iff $v_a(r, p_r) \geq v_a(r', p_{r'})$ for all $r' \in \mathcal{R}$. Note that $(\pi, p)$ is an $\EF$ solution iff $\pi$ is a perfect matching in $\mathcal{F}(p)$. We will refer to $\mathcal{F}(p)$ as the \emph{first-choice graph} at price vector $p \in \mathbb{R}^n$.  

Another useful notion is that of a \emph{linear domain}, which, for a given price vector $p$, specifies a containment-wise maximal region which contains $p$ and wherein the slopes of the utility functions do not change.    

\begin{definition}[Linear Domain]
For a price vector $p \in \mathbb{R}^n$, a linear domain $(L, U) \in \mathbb{R}^n \times \mathbb{R}^n$ is a maximal region $L < p \leq U$ (i.e., $L_r < p_r \leq U_r$ for all $r$) wherein for all agents $a$ and rooms $r$ the utility function $v_a(r, \cdot)$ has a constant slope (left derivative). 
\end{definition}

Note that breakpoints of any utility function lie on the boundary of linear domains, and never in their interior. Since the inequality defining the lower bound $L$ is strict, if for a room $r$ the price $p_r$ is equal to a break point, say $b_i$, then we have $U_r = b_i$ and $L_r < b_i$. Furthermore, the linear domain of a given price vector $p$ can be computed efficiently by finding, for each $p_r$, the breakpoints (among all the breakpoints of $v_a(r, \cdot)$s) that provide the best lower and upper bounds for $p_r$.  

Our algorithms work with a weighted version of the first-choice graph. In particular, for a price vector $p$ with first-choice graph $\mathcal{F}(p)$ and linear domain $(L,U)$, we will use $\mathcal{F}^w(p) :=(\mathcal{A} \cup \mathcal{R}, F, w)$ to denote a bipartite graph with the same edge set, $F$, as $\mathcal{F}(p)$ and edge weights $w_{(a,r)} := \log \lambda^a_r$ for all $(a,r) \in F$. Here, $\lambda^a_r>0$ is the (fixed) slope magnitude of the utility function $v_a(r, \cdot)$ in the linear domain $(L, U)$. 

By definition, only the ``first-choice edges'' are present in $\mathcal{F}^w(p)$ and have weights associated with them. These edge weights (i.e., the logarithm of the slopes) can be negative, since the slopes' magnitudes $\lambda^a_r >0$ are not necessarily greater than one. \\

\noindent
{\bf Approximate Solutions} This paper develops algorithms for efficiently computing solutions that are approximately envy free, i.e., for finding solutions wherein for each agent $a$ the utility of the room assigned to it is multiplicatively close to the utility of the room most preferred by $a$. 

Our algorithm starts with high room rents and iteratively decreases them till an approximately envy free solution is found. Since the rents are high, agents' utilities can be negative during the intermediate steps of the algorithm. The following definition includes this case of negative utilities and thereby provides a unified way to state the approximation guarantees. 

\begin{definition}[Approximately Envy-Free Solutions]
\label{definition:apx-ef}
For a rent-division instance $\langle \mathcal{A}, \mathcal{R}, \{ v_a(r, \cdot) \}_{a,r} \rangle$ and parameter $\varepsilon >0$, a solution $(\pi, p)$ is said to be $\varepsilon$-approximately envy free ($\varepsilon$-$\EF$) iff
\begin{itemize}
\item[(i)] For each agent $a$ that attains nonnegative utility under the solution (i.e., $v_a(\pi(a), p_{\pi(a)}) \geq 0$) the following inequality holds for every room $r \in \mathcal{R}$: $(1+ \varepsilon) v_a(\pi(a), p_{\pi(a)}) \geq v_a(r, p_r)$.
\item[(ii)] For each agent $a$ that attains negative utility under the solution (i.e., $v_a(\pi(a), p_{\pi(a)}) < 0$) the following inequality holds for every room $r \in \mathcal{R}$: $v_a(\pi(a), p_\pi(a)) \geq (1+\varepsilon) v_a(r, p_r)$.
\end{itemize}
\end{definition}

If the base values of the utility functions are positive, $v_a(r,0) >0$, then our algorithm will necessarily find an approximate solution wherein the agents attain nonnegative utilities. Hence, in this standard case, condition $(i)$ of the above definition will hold for every agent under the computed solution, i.e., the solution found by the algorithm will satisfy the usual $(1+\varepsilon)$-approximation guarantee. 

Furthermore, for an $\varepsilon$-$\EF$ solution, $(\pi, p)$, it must be the case that if agent $a$'s utility for the assigned room $\pi(a)$ is negative, then $a$'s utility for every room is negative at the price vector $p$.

\subsection{The Perturbation Lemma of Alkan et al.~\cite{alkan1991fair}}

The following lemma of Alkan et al.~\cite{alkan1991fair} asserts that starting with any fair solution one can always find another envy-free solution with strictly lower prices for all the rooms. In this work we use this lemma for proving the correctness of the developed algorithms although our algorithms do not directly instantiate this result; in fact, in and of itself, the perturbation lemma does not lead to a finite-time algorithm for finding fair solutions; Alkan~\cite{alkan1989existence} provides an instance wherein the perturbations do not converge.


In the following lemma the rents under the two solutions can be negative.

\begin{lemma}[Perturbation Lemma~\cite{alkan1991fair}]
\label{lemma:original-perturb}
Let $(\pi, p)$ be an envy-free solution of a rent-division division instance $\mathcal{I}$. Then, for any small enough $\delta >0$, there exists another envy-free solution $(\sigma, q)$ of $\mathcal{I}$ such that, for all rooms $r$, we have $p_r - \delta \leq q_r < p_r$.   
\end{lemma}

Next we state a version of this lemma which identifies the allocations (bijections) which are realizable in the perturbed solutions. This variant will be used in the analysis of our algorithms---a proof of Lemma~\ref{lemma:perturb} is given in \cref{appendix:perturbation} for completeness.

Recall that $\mathcal{F}^w(p)$ denotes the weighted, first-choice graph at price vector $p$. 

\begin{restatable}{lemma}{LemmaStatementPerturb}
\label{lemma:perturb}
Let $(\pi, p)$ be an envy-free solution of a rent-division instance $\mathcal{I}$ and $\sigma$ be any maximum weight perfect matching in $\mathcal{F}^w(p)$. Then, for any small enough $\delta >0$, there exists another envy-free solution $(\sigma, q)$ of $\mathcal{I}$ such that, for all rooms $r$, we have $ p_r - \delta \leq q_r < p_r $.   
\end{restatable}

Here, $\sigma$ is any maximum weight \emph{perfect} matching in $\mathcal{F}^w(p)$.\footnote{Given that $(\pi, p)$ is $\EF$, $\mathcal{F}^w(p)$ admits a perfect matching (in particular, $\pi$), hence $\sigma$ is well defined.} Since the edge weights can be negative, $\sigma$ is not necessarily a maximum weight matching in the bipartite graph.


\subsection{Optimal Prices}
\label{section:def-opt-prices}
For a rent division instance $\mathcal{I}$, write $\mathcal{E}(\mathcal{I})$ to denote the set of nonnegative prices that induce an envy-free solutions, $\mathcal{E}(\mathcal{I}) := \{ p \in \mathbb{R}_+^n \mid \text{ there exists } \pi \text{ s.t. } (\pi, p) \text{ is } \EF \text{ for } \mathcal{I} \}$. 

The perturbation lemma implies that there are no isolated points in $\mathcal{E}(\mathcal{I})$. The work of Sun and Yang~\cite{sun2003general} provides further insight into the geometry of this set by establishing that $\mathcal{E}(\mathcal{I})$, in fact, contains an \emph{optimal price vector} $p^*$: specifically, Sun and Yang~\cite{sun2003general} show that there exists $p^* \in \mathcal{E}(\mathcal{I})$ such that the componentwise inequality $p^* \leq p$ holds for all $ p \in \mathcal{E}(\mathcal{I})$. 

This, in particular, implies that the total rent imposed under $p^*$ is the lowest among all envy-free solutions with nonnegative prices. One can also prove that (among all fair solutions with nonnegative prices) the sum of agents' utilities is maximized at the envy-free solution $(\pi^*, p^*)$; here, $\pi^*$ is the allocation associated with $p^*$. 

Furthermore, Sun and Yang~\cite{sun2003general} established a notable property of optimal prices in the context of strategic agents: an algorithm  that selects $(\pi^*, p^*)$ as the outcome is guaranteed to be {dominant strategy incentive compatible} (DSIC). 

In \cref{section:optimal-price} we show that if all the slopes of a rent division instance are integer powers of $(1+\varepsilon)$, for parameter $\varepsilon >0$, then we can efficiently find the optimal price. Hence, for such instances not only do we obtain an algorithm for finding $\EF$ solutions, but also an efficient, DSIC mechanism for fair rent division.

\section{Main Results}
\label{section:MainResults}
This section presents the statements of our main results. \\

\noindent
{\bf Exact Algorithm for Structured Instances:} We develop an exact algorithm for rent-division instances in which the slopes of all the piecewise-linear utilities are integer powers of $(1+\varepsilon)$.

\begin{restatable}{theorem}{TheoremEfRounded}
\label{theorem:ef-rounded}
For any given rent-division instance $\overline{\mathcal{I}} = \left\langle \mathcal{A}, \mathcal{R}, \{ \overline{v}_a(r, \cdot) \}_{a,r} \right\rangle$, wherein the utility functions satisfy the powers-of-$(1+\varepsilon)$ property, \cref{algorithm:envy-free} computes an envy-free solution in time that is polynomial in $1/\varepsilon$ and the input size. \\
\end{restatable}

\noindent
{\bf Approximation Algorithms:}
\begin{restatable}{theorem}{TheoremMainResult}
\label{theorem:main-result}
For any given rent-division instance (in which the utility functions are continuous, monotone decreasing, and piecewise linear) an $\varepsilon$-$\EF$ solution can be computed in time polynomial in $1/\varepsilon$ and the input size. 
\end{restatable}

As mentioned previously, there exist rent-division instances wherein negative rents and negative utilities cannot be avoided under any fair distribution of a given rent $C$. Specifically, in the following theorem the computed solution might impose negative rents and utilities. 

\begin{restatable}{theorem}{TheoremFixedCost}
\label{theorem:fixed-cost}
Given a rent-division instance $\mathcal{I}$ along with a total rent $C \in \mathbb{R}$, we can find---in time that is polynomial in $1/\varepsilon$ and the input size---an allocation $\pi$ and a price vector $p \in \mathbb{R}^n$ such that $(\pi, p)$ is an $\varepsilon$-$\EF$ solution of $\mathcal{I}$ and $\sum_r p_r = C$. 
\end{restatable}

The issue of negative utilities is circumvented in prior work (on quasilinear utilities) by adopting the following assumption: $\sum_r z^a_r \geq C$, here $z^a_r$ is the price at which the utility of agent $a$ for room $r$ reduces to zero and $C$ is the total rent. We prove that, even in case of piecewise-linear utilities, under this assumption one can find approximately envy-free solutions with nonnegative utilities.

\begin{restatable}{theorem}{TheoremIr}
\label{theorem:ir}
Let $\mathcal{I}$ be a rent division instance wherein the inequality $\sum_r z^a_r \geq C$ holds for all agents $a$ and parameter $C \in \mathbb{R}$. Then, we can find---in time that is polynomial in $1/\varepsilon$ and the input size---an allocation $\pi$ and a price vector $p \in \mathbb{R}^n$ such that (i) $(\pi, p)$ is an $\varepsilon$-$\EF$ solution of $\mathcal{I}$, (ii) $\sum_r p_r = C$, and (iii) the utilities of all the agents under $(\pi, p)$ are nonnegative. \\
\end{restatable}

\noindent
{\bf DSIC Mechanism for Structured Instances:} We also design an algorithm for finding optimal solutions of instances in which the slopes of all the utilities are integer powers of $(1+\varepsilon)$.  Hence, using the result of Sun and Yang~\cite{sun2003general}, we obtain a DSIC mechanism for this setting.

\begin{restatable}{theorem}{TheoremPostpro} 
\label{theorem:postpro}
Given any rent-division instance $\overline{\mathcal{I}} = \left\langle \mathcal{A}, \mathcal{R}, \{ \overline{v}_a(r, \cdot) \}_{a,r} \right\rangle$, wherein the utility functions satisfy the powers-of-$(1+\varepsilon)$ property, \cref{algorithm:opt-price} finds an optimal envy-free solution with runtime polynomial in $1/\varepsilon$ and the input size. \\
\end{restatable}

\noindent
{\bf Complexity of Fair Rent Division:} $\plc$ refers to the total search problem of finding an envy-free solution of any rent-division instance with continuous, monotone decreasing, and piecewise-linear utilities. We show that 

\begin{restatable}{theorem}{TheoremComplexity}
\label{theorem:complexity}
$\plc$ is in $\PPAD \cap \PLS$.
\end{restatable}

\section{Exact Algorithm for Structured Instances}
\label{section:rounded}

This section considers instances wherein the constituent slopes of all the utility functions are integer powers of $(1 + \varepsilon)$, for a fixed $\varepsilon >0$. Specifically, in such instances, for all agents $a$, rooms $r$, and pieces $i$, we have $\lambda^a_{r,i} = (1 + \varepsilon)^k$, for some integer $k$ (which can be negative and depends on $a$, $r$, and $i$). In the next section we will extend the analysis to general instances. 


Write $\overline{v}_a(r, \cdot)$ to denote piecewise-linear utility functions wherein the slopes satisfy this powers-of-$(1+\varepsilon$) property, for a fixed $\varepsilon >0$. We will show that, given a rent division instance $\langle \mathcal{A}, \mathcal{R}, \{ \overline{v}_a(r, \cdot) \}_{a,r} \rangle$, Algorithm~\ref{algorithm:envy-free} ($\ALG$) finds an envy-free solution in time polynomial in $1/\varepsilon$ and the bit complexity of the input. 


Below we will prove that $\ALG$ can, in fact, find an envy-free solution for an arbitrary  rent-division instance $\mathcal{I}=\langle \mathcal{A}, \mathcal{R}, \{ {v}_a(r, \cdot) \}_{a,r} \rangle$. It is the runtime analysis of the algorithm that requires the powers-of-$(1+\varepsilon)$ property. 

Given instance $\mathcal{I}$, $\ALG$ begins by considering modified utility functions, $\widehat{v}_a(r, \cdot)$s, that match ${v}_a(r, \cdot)$s till a large  threshold $M$ and for prices higher than $M$ the modified utility functions are quasilinear. In particular, let $M \in \mathbbm{R}_+$ be such that ${v}_a(r, M) < {v}_a(r',0)$, for all agents $a \in \mathcal{A}$ and rooms $r, r' \in \mathcal{R}$. Since the utility functions are monotone decreasing, such an $M$ exists and can be computed efficiently.\footnote{We can conservatively set $M =\frac{\max_{a,r} {v}_a(r,0) - \min_{a',r'} {v}_{a'}(r',0) }{ \min_{a, r, i} \lambda^a_{r,i} }$} 

For all agents $a$ and rooms $r$, define  
 \begin{align}
 \label{equation:bar-to-hat}
    \widehat{v}_a(r, x ) := 
    \begin{cases}
     {v}_a(r, x)  & \text{ for }  x \leq M \\
        {v}_a(r, M) - (x - M) & \text { for } x >M
    \end{cases}
\end{align}

{	
\begin{algorithm}[ht]
		{
		{\bf Input:} A rent-division instance $\mathcal{I} = \left\langle \mathcal{A}, \mathcal{R}, \{ {v}_a(r, \cdot) \}_{a,r} \right\rangle$ with continuous, monotone decreasing, piecewise-linear utility functions \\ 
		{\bf Output:} An envy-free solution of $\mathcal{I}$
		\caption{$\ALG$: Algorithm for fair rent division}
		\label{algorithm:envy-free}
		\begin{algorithmic}[1]
			\STATE For each agent $a \in \mathcal{A}$ and room $r \in \mathcal{R}$, construct valuation $\widehat{v}_a(r, \cdot)$ as detailed in equation (\ref{equation:bar-to-hat})
			\STATE Compute an envy-free solution $(\pi^0, p^0)$ for $\widehat{\mathcal{I}} := \left\langle \mathcal{A}, \mathcal{R}, \{ \widehat{v}_a(r, \cdot) \}_{a,r} \right\rangle$ with $p^0 \geq M \mathbb{1}$  \label{step:hat-instance} \\
			\COMMENT{Such a solution can be found efficiently using, say, the algorithm given in \cref{section:quasi}} 
			\STATE Initialize $i \leftarrow 0$ and let $(L^0,U^0) \in \mathbbm{R}^n \times \mathbbm{R}^n$ be the linear domain containing $p^0$
			\WHILE{$p^i_r >0$ for all rooms $r \in \mathcal{R}$}
			\STATE Update $i \leftarrow i +1$
			\STATE Set $\pi^i$ to be the maximum weight \emph{perfect} matching in $\mathcal{F}^w(p^{i-1})$
			\STATE Set $p^i$ to be the optimal solution of the following linear program
			\begin{align*}
			\min_{x \in \mathbbm{R}^n} & \ \  \sum_r x_r \tag{LP$(L^{i-1}, U^{i-1}, \pi^i)$}\\  
			\text{subject to} & \ \ \ U^{i-1} \geq x \geq L^{i-1}\\ 
			& \ \ \ \widehat{v}_a( \pi^i(a) , x_{\pi^i(a)}) \geq \widehat{v}_a(r, x_r) \quad \text{ for all } a \in \mathcal{A}  \text{ and } r \in \mathcal{R} \\
			& \ \ \ x  \geq \mathbb{0} 
			\end{align*}
			\STATE Set $(L^i,U^i)$ to be the linear domain of $p^i$ \\
			\COMMENT{The linear domain changes iff one of the inequalities in $p^i \geq L^{i-1}$ becomes tight}
			\ENDWHILE
			\STATE Return $(\pi^i, p^i)$
		\end{algorithmic}
	}
	\end{algorithm}
}

Recall that envy-free solution for quasilinear utilities can be computed efficiently. Therefore, we can find a fair solution $(\pi^0, p^0)$ for the quasilinear functions $\{{v}_a(r, M) - (x - M)\}_{a,r}$ . Furthermore, by adding $M$ to the (nonnegative) price of each room we can ensure that $(\pi^0, p^0)$ is not only envy free, but also satisfies $p^0_r \geq M$ for all $r$.

By construction, $(\pi^0, p^0)$ is $\EF$ for $\widehat{\mathcal{I}} = \langle \mathcal{A}, \mathcal{R}, \{ \widehat{v}_a(r, \cdot) \}_{a, r} \rangle$. $\ALG$ starts with such a solution and iteratively reduces the prices, while maintaining envy freeness with respect to $\widehat{v}_a(r, \cdot)$s. The algorithm terminates when it finds an envy-free solution $(\pi, p)$ wherein the rent for some room has been reduced to zero. One can show that at this price vector $p$ the rents of all the rooms have to be less than $M$ (details of this argument appear in the proof of Theorem~\ref{theorem:ef-rounded}). Since all the components of $p$ are less than $M$, at this price vector the utilities under the functions $\{ \widehat{v}_a(r, \cdot)\}_{a, r}$ and $\{ {v}_a(r, \cdot)\}_{a, r}$ are equal. Therefore, $(\pi, p)$ is $\EF$ with respect to ${v}_a(r, \cdot)$s as well. Intuitively, this establishes the correctness of the algorithm. Below we analyze the runtime of $\ALG$ and, overall, show that this algorithm efficiently finds an envy-free solution. 


The next few lemmas establish useful properties of the intermediate solutions, $(\pi^i, p^i)$s, computed by the algorithm. 

\begin{lemma}
\label{lemma:ef-invariant}
Given a rent-division instance $\mathcal{I}$, let $\widehat{\mathcal{I}}$ denote the modified instance obtained in Step~\ref{step:hat-instance} of $\ALG$ and $(\pi^i, p^i)$ be the solution computed in the $i$th iteration of the algorithm. Then, $(\pi^i, p^i)$ is an envy-free solution for $\widehat{\mathcal{I}}$. 
\end{lemma}
\begin{proof}
We will prove this claim by induction over $i$. $\ALG$ starts with a solution $(\pi^0, p^0)$ which is envy free with respect to the quasilinear functions $\{{v}_a(r, M) - (x - M)\}_{a,r}$ and satisfies $p^0_r \geq M$ for all $r$. For such high rents the utilities under $\widehat{v}_a(r, \cdot)$s are equal to the ones under the quasilinear functions $\{{v}_a(r, M) - (x - M)\}_{a,r}$.  Hence, for the base case, we have that $(\pi^0, p^0)$ is $\EF$ for $\widehat{\mathcal{I}}= \left\langle \mathcal{A}, \mathcal{R}, \{ \widehat{v}_a(r, \cdot) \}_{a,r} \right\rangle$.

By the induction hypothesis, solution $(\pi^{i-1}, p^{i-1})$ is $\EF$ for $\widehat{\mathcal{I}}$. Therefore, $\pi^{i-1}$ is a perfect matching in the first-choice graph $\mathcal{F}(p^{i-1})$. This, in turn, implies that $\pi^i$---which is set to be a maximum weight perfect matching in $\mathcal{F}^w(p^{i-1})$---is well defined. 

The linear program LP$(L^{i-1}, U^{i-1}, \pi^i)$---which is solved in $\ALG$ to obtain $p^i$---considers all the price vectors (in the current linear domain $(L^{i-1}, U^{i-1})$) under which $\pi^i$ is an envy-free allocation.  
Since this linear program is bounded and feasible,\footnote{The price vector $p^{i-1}$, in particular, is a feasible solution of this linear program.} we get that it has an optimal solution, $p^i$. Overall, these observations imply that  $(\pi^i, p^i)$ is $\EF$ for $\widehat{\mathcal{I}}$, and the claim follows. 
\end{proof}

Given a price vector $p$, write $w(\pi)$ to denote the weight of the a matching $\pi$ in the weighted, bipartite graph $\mathcal{F}^w(p)$. Recall that for each edge $(a, r)$ in this graph, the edge weight $w_{(a, r)}$ is equal to $\log \lambda^a_r$, where $\lambda^a_r$ is the slope of the utility function $\widehat{v}_a(r, \cdot)$ in the linear domain containing $p$. Also, we have that the slopes remain unchanged in a linear domain. Hence, if the linear domain of price vector $p'$ is the same as the linear domain of $p$, then the weight of a matching $\pi$ in $\mathcal{F}^w(p')$ is equal to its weight in $\mathcal{F}^w(p)$.

Whenever we will compare weights of matchings the linear domain will be fixed. Hence, for ease of presentation, we will not explicitly denote the dependence of the weights, $w(\cdot)$, on the underlying linear domain. 
 
\begin{lemma}
\label{lemma:wts-drop}
Let $(\pi^{i}, p^{i})$ be a solution computed by $\ALG$ in an iteration $i$ wherein the linear domain does not change, $(L^{i-1}, U^{i-1}) =(L^i, U^i)$. Then,   the following strict inequality holds $w(\pi^i) < w(\pi^{i+1})$. \\ Here, $\pi^{i+1}$ is the allocation computed by $\ALG$ in the $(i+1)$th iteration, i.e., $\pi^{i+1}$ is a maximum weight perfect matching in $\mathcal{F}^w(p^i)$.  
\end{lemma}

\begin{proof}
$\ALG$ computes $p^i$ by solving the linear program LP$(L^{i-1}, U^{i-1}, \pi^i)$ which, in particular, includes constraints of the form $x \geq L^{i-1}$ and $x  \geq \mathbb{0}$. 
Note that if any one of these constraints is tight for $p^i$, then either the linear domain changes\footnote{Recall that the definition of a linear domain, $(L,U)$, for a price vector $p$ mandates a strict inequality between $L$ and $p$, i.e., it requires $L < p$.} (i.e., $(L^{i-1}, U^{i-1})  \neq (L^i, U^i)$) or the while-loop terminates. Therefore, $p^i$ must satisfy the following strict inequalities $p^i > L^{i-1}$ and $p^i > \mathbb{0}$. 

Lemma~\ref{lemma:ef-invariant} ensures that $(\pi^{i}, p^{i})$ is $\EF$ and, hence, $\pi^{i}$ is a perfect matching in the bipartite graph $\mathcal{F}^w(p^{i})$. Using the fact that $\pi^{i+1}$ is a maximum weight perfect matching in this graph, we get $w(\pi^{i+1}) \geq w(\pi^{i})$. Next we will prove that this inequality is never tight. 

Assume, for contradiction, that $w(\pi^i) = w(\pi^{i+1})$, i.e., $\pi^{i}$ is also a maximum weight perfect matching in $\mathcal{F}^w(p^{i})$.
In this case, the perturbation lemma (Lemma~\ref{lemma:perturb}) implies that, for any small enough $\delta >0$, there exists a price vector $q$ such that $(\pi^i, q)$ is envy free and $p^i_r - \delta \leq q_r < p_r$ for all $r$. 

Since we have strict inequalities $p^i > L^{i-1}$ and $p^i_r > 0$, an appropriate $\delta >0$ leads to a price vector $q$ which satisfies $q \geq L^{i-1}$ and $q \geq \mathbb{0}$. In other words, $q$ is a feasible solution of the linear program LP$(L^{i-1}, U^{i-1}, \pi^i)$, whose optimal solution is $p^i$.\footnote{Here, even if the inequality $p^i_{r'} \leq U^{i-1}_{r'}$  is tight, for some $r'$, $q$ remains feasible.} The inequality $\sum_r q_r < \sum_r p^i_r$ contradicts the optimality of $p^i$ and establishes the stated claim. 
\end{proof}

The following lemma shows that the room rents do not increase as the algorithm progresses. 

\begin{lemma}
\label{lemma:price-drop-sale-sale}
The price vectors, $p^{i-1}$ and $p^{i}$, computed in consecutive iterations of $\ALG$ satisfy the following componentwise inequality: $p^{i-1} \geq p^{i}$.
 \end{lemma}
\begin{proof}
The price vector $p^{i-1}$ is a feasible solution of the linear program LP$(L^{i-1}, U^{i-1}, \pi^i)$: by definition, $p^{i-1}$ is contained in the linear domain $(L^{i-1}, U^{i-1})$ and $\pi^i$ is a maximum weight perfect matching in $\mathcal{F}^w(p^{i-1})$, i.e., $(\pi^i, p^{i-1})$ is $\EF$. 

We will show that the inequality $p \geq p^i$ holds for all price vectors $p \in \mathbbm{R}_+^n$ that are feasible with respect to the linear program LP$(L^{i-1}, U^{i-1}, \pi^i)$. Hence, the feasibility of $p^{i-1}$ will gives us the desired inequality $p^{i-1} \geq p^{i}$.  

Consider a feasible solution, $p \in \mathbbm{R}_+^n$, of the linear program LP$(L^{i-1}, U^{i-1}, \pi^i)$. Say, for contradiction, that there exists a room $r$ such that $p_r < p^i_r$. We will show that the price vector, say $q$, obtained by taking the componentwise minimum of $p^i$ and $p$ is also a feasible solution. Note that such a $q$ would satisfy $\sum_r q_r < \sum_r p^i_r$, which would contradict the fact that $p^i$ is an optimal solution of LP$(L^{i-1}, U^{i-1}, \pi^i)$. Hence, it must be the case that $p \geq p^i$ for all feasible price vectors $p$. 

Write $q_{r'} := \min\{ p_{r'}, p^i_{r'} \}$ for all $r' \in \mathcal{R}$. This price vector directly satisfies the box constraints of the linear program: $L^{i-1} \leq q \leq U^{i-1}$ and $q \geq \mathbb{0}$. To establish the feasibility of $q$ it remains to show that $q$ maintains the envy freeness of $\pi^i$, i.e., it satisfies $\widehat{v}_a( \pi^i(a) , q_{\pi^i(a)}) \geq \widehat{v}_a(r, q_r)$ for all $a$ and $r$. 

For each room $r$, there are two possible cases either $q_r = p_r$ or $q_r = p^i_r$. If $q_r = p_r$, then we have $ \widehat{v}_a(r, q_r) = \widehat{v}_a(r, p_r)$ for all agents $a$. The feasibility of $p$ ensures that $\widehat{v}_a( \pi^i(a) , p_{\pi^i(a)}) \geq \widehat{v}_a(r, p_r)$. Since, $q_{\pi^i(a)} \leq p_{\pi^i(a)}$, we get the desired inequality $\widehat{v}_a( \pi^i(a) , q_{\pi^i(a)}) \geq \widehat{v}_a( \pi^i(a) , p_{\pi^i(a)}) \geq \widehat{v}_a(r, p_r) = \widehat{v}_a(r, q_r)$. The inequality analogously holds for the case wherein $q_r = p^i_r$. This, overall, proves that $q$ is feasible and the lemma follows. 
\end{proof}

Using the lemmas mentioned above, we will now prove the main result of this section. 

\TheoremEfRounded*


\begin{proof}
Consider a sequence of successive iterations $S=\{i, i+1, i+2, \ldots, j\}$ of $\ALG$ in which the underlying linear domain remains the same, i.e., $(L^{i-1}, U^{i-1}) = (L^i, U^i) = \cdots = (L^j, U^j)$. We will first upper bound the length, $|S|=(j+1-i)$, of any such sequence. 

Write $\lambda^a_r$ to denote the (fixed) slope of the utility function $\widehat{v}_a(r, \cdot)$ in the linear domain $(L^{i-1}, U^{i-1})$. Since all the slopes in the given piecewise-linear utilities are integer powers of $(1 + \varepsilon)$, we have $\lambda^a_r= (1+\varepsilon)^{k^a_r}$ for an integer $k^a_r \in \mathbbm{Z}$. The slopes, $\lambda^a_r$s, are fixed throughout the linear domain and, hence, the exponents, $k^a_r$s, are fixed as well.   Note that $k^a_r$ can be negative, but its magnitude is upper bounded in terms of  the bit complexity of $\lambda^a_r$: using $|k^a_r| = \frac{|\ln \lambda^a_r|}{\ln (1 + \varepsilon)}$ and $\ln (1+ \varepsilon) \geq \varepsilon - \varepsilon^2/2$, we obtain  $|k^a_r| \leq \frac{3 |\log \lambda^a_r|}{\varepsilon}$, for $\varepsilon <1$. Therefore, if $\beta$ is the bit complexity of the input parameters, then $|k^a_r| \leq \frac{3\beta}{\varepsilon}$.

This magnitude bound implies that $| \max_{a, r} k^a_r  - \min_{a, r} k^a_r| \leq \frac{6 \beta}{\varepsilon}$. Furthermore, the edge weights, $w_{(a,r)}$s, in the bipartite graph $\mathcal{F}^w(p^k)$, for any $k \in S$, satisfy $w_{(a, r)} = k^a_r \ \log (1 + \varepsilon)$. Therefore, the difference in the weights of any two perfect matchings in the complete bipartite graph $(\mathcal{A} \cup \mathcal{R}, \mathcal{A} \times \mathcal{R})$---with edge weights set to $w_{(a, r)}$s---is upper bounded by $\frac{6\beta}{\varepsilon} \ n \log (1 + \varepsilon) $. 

Since the linear domain remains unchanged throughout the sequence $S$, Lemma~\ref{lemma:wts-drop} implies the following strict inequalities $w(\pi^i) < w(\pi^{i+1}) < \ldots < w(\pi^{j+1})$; here $\pi^k$ is the matching computed by $\ALG$ in the $k$th iteration, for $k \in S$. The above mentioned observation gives us
\begin{align}
\label{ineq:wt-diff}
w(\pi^{j+1}) - w(\pi^i) \leq \frac{6\beta}{\varepsilon} \ n \log (1 + \varepsilon). 
\end{align}

Furthermore, for each $k \in S$, the weight $w(\pi^k)$ is an integer multiple of $\log (1 + \varepsilon)$; recall that each exponent $k^a_r$ is an integer. Therefore, using the strict inequality $w(\pi^k) < w(\pi^{k+1})$, we get that $w(\pi^{k+1}) - w(\pi^k) \geq \log (1 + \varepsilon)$. This inequality and equation (\ref{ineq:wt-diff}) bound the length of the sequence 
\begin{align}
\label{ineq:seq-length}
|S| = j +1 -i & \leq  \frac{6\beta n}{\varepsilon} 
\end{align} 

This bound implies that $\ALG$ changes the underlying linear domain after at most $ \frac{6\beta n}{\varepsilon}$ successive iterations. We will complete the runtime analysis of $\ALG$ by showing that the total number of linear-domain changes in $\ALG$ is polynomially bounded.  

Let $\ell^a_r$ denote the total number of pieces in the piecewise-linear utility function $\widehat{v}_a(r, \cdot)$. Write $\ell := \sum_{a, r} \ell^a_r$. Consider counters $c^a_r \in \mathbbm{Z}_+$ and to keep track of the price changes we adopt the following convention: after the $i$th iteration, with price vector $p^i$ in hand, the counter $c^a_r$ is set to be the index of the piece of $\widehat{v}_a(r, \cdot)$ which contains $p^i_r$; in particular, if the price $p^i_r$ is between the $t$th and $(t+1)$th breakpoint of $\widehat{v}_a(r, \cdot)$, then at the $i$th iteration $c^a_r = t$. Since the initial price vector satisfies $p^0_r \geq M$, for all $r$, and $M$ is the last breakpoint of all the utility functions $\widehat{v}_a(r, \cdot)$s, we get that $c^a_r$ is initialized to $\ell^a_r$ for each $a$ and $r$. That is, at the beginning of the algorithm we have $\sum_{a, r} c^a_r = \ell$. 

Lemma~\ref{lemma:price-drop-sale-sale} ensures that the computed prices are nonincreasing. Hence, the counters $c^a_r$s are nonincreasing as well. Furthermore, in $\ALG$, a change in the linear domain is triggered iff, for some $r$, the inequality $p^i_r \geq L^{i-1}_r$ becomes tight. By the definition of a linear domain we get that there exists an agent $a'$ such that $L^{i-1}_r$ is equal to a breakpoint of $\widehat{v}_{a'}(r, \cdot)$. Therefore, in $\ALG$, whenever the linear domain changes the value of at least one counter, $c^{a'}_r$, gets decremented by one. As mentioned above, the counter values do not increase as $\ALG$ progresses and initially they satisfy $\sum_{a, r} c^a_r = \ell$. Also, note that $\ALG$ terminates as soon as the price of a room reduces to zero, hence the counter values satisfy $c^a_r \geq 1$ during the entire execution of the algorithm. These arguments prove that the total number of domain changes is at most $\ell$. Using inequality (\ref{ineq:seq-length}) with this bound we get that $\ALG$ runs for at most $\frac{6 \beta n \ell}{\varepsilon}$ iterations.\footnote{Since $\beta$ is the bit complexity of the input parameters and $\ell$ is the total number of given pieces, the input size is $\mathcal{O}(\beta \ell n)$ and, hence, an algorithm that runs in time polynomial in $\beta$, $\ell$, and $n$ is deemed to be efficient.} This completes the runtime analysis of the algorithm.

To complete the proof we will show that the solution returned by $\ALG$, say $(\pi, p)$, is $\EF$ for the original instance $\mathcal{I}$. This solution is $\EF$ for $\widehat{\mathcal{I}}$ (Lemma~\ref{lemma:ef-invariant}) and the termination condition of the while-loop ensures that $p_{\rho} = 0$ for some room $\rho$. Hence, $p$ is componentwise less than $M$: assume, for contradiction, that $p_{r'} > M$ for a room $r'$ and let $\alpha$ be the agent who is assigned this room, $\pi(\alpha) = r'$. The following inequalities contradict the envy freeness of $(\pi, p)$ for $\widehat{\mathcal{I}}$

\begin{align*}
\widehat{v}_{\alpha}(\pi(\alpha), p_{\pi(\alpha)}) & < \widehat{v}_{\alpha}(\pi(\alpha), M)   \qquad \text{(utility $\widehat{v}_{\alpha}(\pi(\alpha), \cdot)$ is monotone decreasing)} \\
& = v_\alpha(\pi(\alpha), M) \qquad \text{(follows from (\ref{equation:bar-to-hat}))} \\
& < v_\alpha(\rho, 0) \qquad \text{(by the definition of $M$)} \\
& = \widehat{v}_\alpha(\rho, 0) \qquad \text{(again, follows from (\ref{equation:bar-to-hat}))} \\
& = \widehat{v}_\alpha(\rho, p_\rho).
\end{align*}

Therefore, the returned price vector satisfies $p \leq M \mathbb{1}$. Note that, when all the prices are less than $M$,  the equality $\widehat{v}_a(r, p_r) = v_a(r, p_r)$ holds for all $a$ and $r$  (see equation (\ref{equation:bar-to-hat})). Overall, we get that $(\pi, p)$ is an envy-free solution of $\mathcal{I}$ as well, and this completes the proof.  
\end{proof}

\subsection{Runtime Analysis for Special Cases}
\label{subsection:special}

\cref{theorem:ef-rounded} establishes the time complexity of $\ALG$ (\cref{algorithm:envy-free}) for rent-division instances in which all the slopes are integer powers of $(1+\varepsilon)$. Specifically, this runtime analysis shows that the number of iterations in $\ALG$ that consider the same linear domain $(L, U)$ is at most the number of perfect matchings with \emph{distinct} weights in the first-choice graph associated with $(L, U)$; see inequality (\ref{ineq:seq-length}). Write $D$ to denote the maximum number of such perfect matchings.  The powers-of-$(1+\varepsilon)$ property is used, in particular,  to show that for  structured instances $D = \mathcal{O}(\frac{\beta n}{\epsilon})$, which in turn shows that  any linear domain is considered in $\mathcal{O}(\frac{\beta n}{\epsilon})$ iterations; here, $\beta$ is the bit complexity of the input. 

For general instances, $D$ continues to upper bound the number of iterations which consider the same linear domain. However, without the powers-of-$(1+\varepsilon)$ property $D$ can be as large as $n!$. 

Next, we present two special cases of piecewise-linear utilities in which we can prove tighter bounds on $D$ and, hence, establish the efficiency of $\ALG$. \\

\noindent 
{\bf Constant Number of Agents:} Consider the case in which $n$ is a fixed constant. This setting occurs naturally in real-world applications wherein the number of roommates is small. Since $D$ is at most $n!$, in this case $D$ is also a constant. Therefore, when the number of agents is fixed,  $\ALG$ finds an envy-free solution in polynomial time.  \\

\noindent 
{\bf Constant Number of Distinct Slopes:} For a rent-division instance, let $k$ denote the total number of distinct slope values across all the utility functions. Equivalently, $k$ is the total number of distinct logarithmic values of the slopes. Note that the total number of pieces over all the utility functions can be much larger than $k$. 

For a fixed linear domain, consider the complete bipartite graph in which the edge weights are equal to the logarithm of the slopes. Since there are no more than $k$ distinct edge weights, say $w_1, w_2, \ldots, w_k$, with each perfect matching $\mu$ in the complete bipartite graph we can associate a $k$-tuple $(n_1,n_2,n_3,....,n_k)$, where $n_i$ represents the number of edges in $\mu$ whose weight equal to $w_i$. Note that the weight of the perfect matching $\mu$ is equal to $\sum_{i=1}^k n_i w_i$, i.e., the tuple uniquely determines the weight of $\mu$. Also, for any matching $\sum_{i=1}^k n_i = n$. In other words, the tuple is a $k$ integer partition of $n$. Therefore, the number of tuples is at most $(n+1)^k$. This shows that $D$, the number of perfect matchings with distinct weights, is also upper bounded by $(n+1)^k$. Hence, if $k$ is a fixed constant, then $\ALG$ finds an envy-free solution in polynomial time.

This case (with $k=2$) captures the setting in which the agents have quasilinear utilities and a budget constraint (i.e., an agent cannot pay more than a specified amount as rent). A polynomial-time algorithm for this setting was developed in \cite{procacciafair}. Since this budgeted setting can be modeled by piecewise-linear utilities that have two distinct slope values (the slope of the quasilinear part, $-1$, and a large negative slope at the budget of every agent), this work also obtains an efficient algorithm for finding fair solutions under quasilinear utilities and budget constraints.

The analysis for the two special cases above also hold for the modified versions of \cref{algorithm:envy-free} that are developed in \cref{section:fixed} and \cref{section:optimal-price}.

\section{Approximation Algorithm for Fair Rent Division}
\label{section:FPTAS}
This section presents a fully polynomial-time approximation scheme (FPTAS) for the fair rent-division problem under continuous, monotone decreasing, and piecewise-linear utilities. Given a rent-division instance and an approximation factor $\varepsilon$, our first step is to construct a ``close-by'' structured instance in which the slopes of all the utility functions are integer powers of $(1+\varepsilon)$. We will prove that any envy-free solution of this constructed instance is an $\varepsilon$-$\EF$ solution for the original instance. Hence, using the algorithm developed in Section~\ref{section:rounded} ($\ALG$) we obtain an FPTAS. 

Given a problem instance $\mathcal{I} = \left\langle \mathcal{A}, \mathcal{R}, \{ {v}_a(r, \cdot) \}_{a,r} \right\rangle$, we construct utility functions $\{ \overline{v}_a(r, \cdot) \}_{a,r}$ by rounding the slopes of the functions $\{ {v}_a(r, \cdot) \}_{a,r}$ piece by piece. 


Consider piece $i$ of the utility function $v_a(r,x)$ and let $x_1$ and $x_2$ be the breakpoints corresponding to this piece (i.e., the piece is defined in the interval $[x_1,x_2]$). Write $\lambda^a_r$ to denote the magnitude of the slope of $v_a(r, \cdot)$ in this piece. We will use $H^a_r$ and $B^a_r$ to denote that value of the utility function $v_a(r, \cdot)$ at breakpoints $x_1$ and $x_2$, respectively:  $H^a_r = v_a(r,x_1)$ and $B^a_r = v_a(r,x_2)$. Since the construction is identical for all pieces, we will overload the notation and not explicitly index these parameters by $i$; each piece of $v_a(r, \cdot)$ will be denoted by $ \left\langle [x_1,x_2], \lambda^a_r, H^a_r, B^a_r \right\rangle$. 


The following equalities hold for all $x \in [x_1, x_2]$
\begin{align}
\label{eq:rewrite}
v_a(r,x) = H^a_r - \lambda^a_r \ (x - x_1 )  = B^a_r + \lambda^a_r \ (x_2 - x ) 
\end{align}


Given a piece $ \left\langle [x_1,x_2], \lambda^a_r, H^a_r, B^a_r \right\rangle$ and the approximation parameter $\varepsilon \in (0,1)$, we round up $\lambda^a_r$ to the nearest integer power of $(1+\varepsilon)$: define $\overline{\lambda^a_r} \coloneqq (1+\varepsilon)^{\ceil*{\log_{(1+\varepsilon)}\lambda^a_r}}$ and note that 

\begin{align}
\label{property:sandwich}
 \frac{\overline{\lambda^a_r}}{1+\varepsilon} < \lambda^a_r \leq \overline{\lambda^a_r}
\end{align}


To obtain the approximating function $\overline{v}_a(r, \cdot)$, each piece of $v_a(r,x)$ is substituted by two linear pieces with slopes $\overline{\lambda^a_r}$ and $\frac{\overline{\lambda^a_r}}{1+\varepsilon}$, respectively. Note that, by the intermediate value theorem, there always exists an $x^* \in [x_1,x_2)$ such that $H^a_r = B_a^r + \overline{\lambda^a_r}(x_2-x^*) + \frac{\overline{\lambda^a_r}}{1+\varepsilon}(x^*-x_1)$.

\begin{figure}[h]
\begin{center}
\includegraphics[scale=.52]{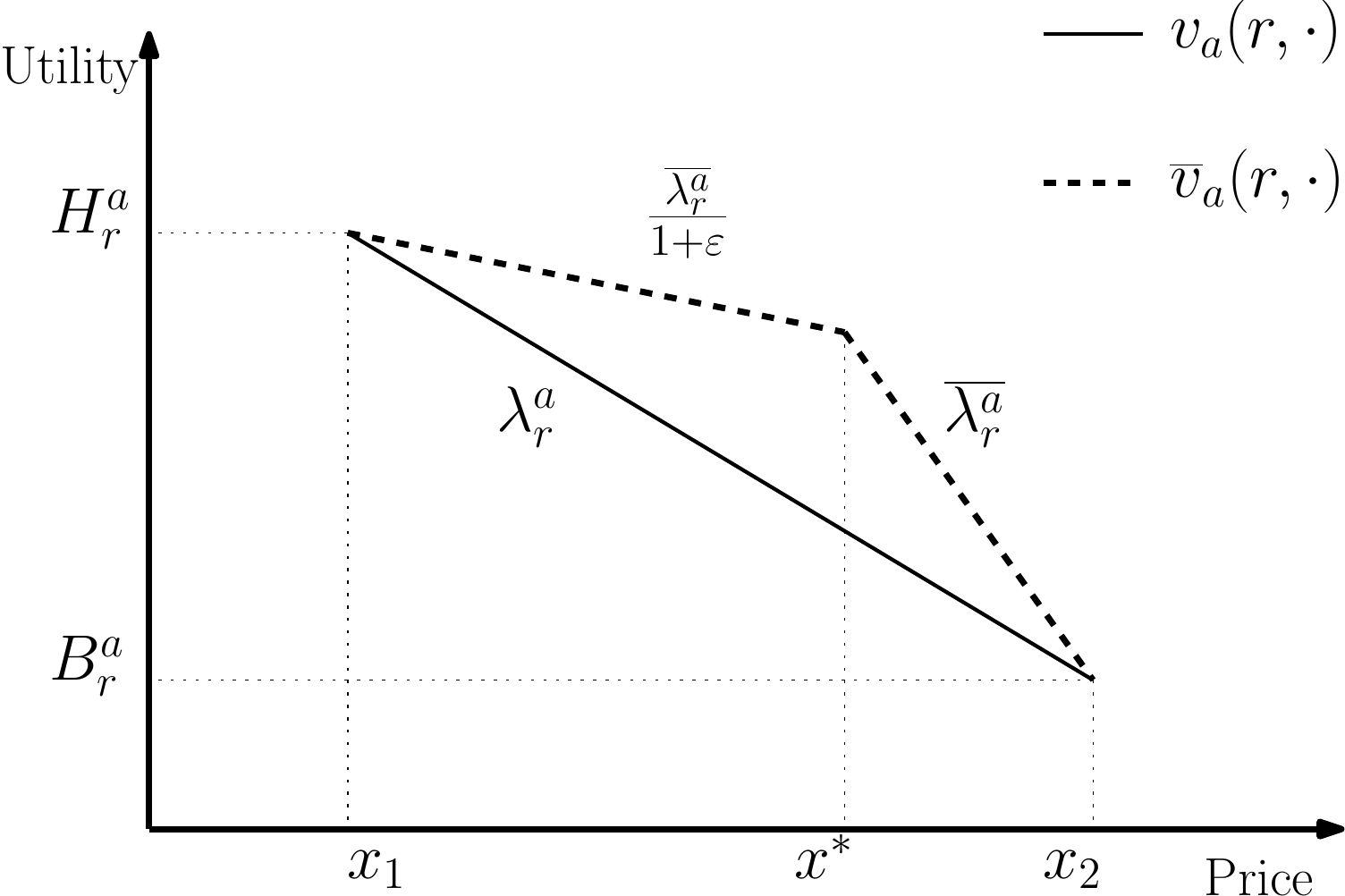}
\end{center}
\caption{Construction of $\overline{v}_a(r, \cdot)$ in a single piece.}
\label{figure:rounding}
\end{figure}

Using $x^*$ (which can be found by solving the equation above), we define $\overline{v}_a(r,x)$ in the interval $[x_1, x_2]$ as follows (see ~\cref{figure:rounding})

\begin{align*}
\overline{v}_a(r,x) := \begin{cases}
H_a^r - \frac{\overline{\lambda^a_r}}{1+\varepsilon}(x - x_1)  & \text{ for }x_1 \le x < x^* \\
B^a_r + \overline{\lambda^a_r}(x_2-x) & \text{ for }x^* \le x \le x_2 
\end{cases}
\end{align*}

The piecewise-linear functions $\{\overline{v}_a(r, \cdot) \}_{a,r}$ can be found in {polynomial} time by applying the above mentioned procedure over the pieces of the given utilities $\{{v}_a(r, \cdot) \}_{a,r}$.

For all agents $a$ and rooms $r$, let $z^a_r$ be the price at which agent $a$'s utility for room $r$ reduces to zero, i.e., $z^a_r$ is the solution of $v_a(r, z^a_r) = 0$. While constructing each $\overline{v}_a(r, \cdot)$, we will treat $z^a_r$ as a breakpoint, i.e., if $z^a_r$ is contained in the piece $[x_1,x_2]$, then $[x_1,z^a_r]$ and $[z^a_r,x_2]$ will be treated as distinct pieces during the construction of $\overline{v}_a(r, \cdot)$. This convention ensures the following property

\begin{lemma}
	\label{lemma:same-sign}
	For all agents $a \in \mathcal{A}$, rooms $r \in \mathcal{R}$ and any price $x\in \mathbbm{R}$, we have $\overline{v}_a(r,x) \geq 0$ iff $v_a(r,x) \geq 0$. In addition, for a price $z \in \mathbbm{R}$ the equality $\overline{v}_a(r, z) = 0$ holds iff $v_a(r, z) = 0$.
\end{lemma}

The construction of  $\overline{v}$ directly implies the following properties (see Figure~\ref{figure:rounding}).

\begin{lemma}
\label{lemma:properties}
For every agent $a \in \mathcal{A}$, room $r \in \mathcal{R}$, and each linear piece $ \left \langle [x_1,x_2], \lambda^a_r, H^a_r, B^a_r \right \rangle$ of $v_a(r, \cdot)$, the following hold for all $x \in [x_1,x_2]$

\begin{enumerate}[(a)]
	\item $\overline{v}_a(r,x) \geq v_a(r,x)$
	\item $B^a_r + \overline{\lambda^a_r} (x_2 - x)  \geq \overline{v}_a(r,x)$
	\item $H^a_r - \frac{\overline{\lambda^a_r}}{1+\varepsilon} (x - x_1)  \geq  \overline{v}_a(r,x)$
	\item $\overline{v}_a(r,x) \geq 0$ iff $H^a_r \ge 0$ and $B^a_r \geq 0$
	\item $\overline{v}_a(r,x) \leq 0$ iff $H^a_r \le 0$ and $B^a_r \leq 0$
\end{enumerate}
\end{lemma}

Note that at each breakpoint $b_i$ of the given utility function $v_a(r,\cdot)$ we have $v_a(r, b_i) = \overline{v}_a(r, b_i)$. This observation and the construction of $\{\overline{v}_a(r, \cdot) \}_{a,r}$ (in each piece) imply that these functions are continuous and monotone decreasing. Formally, 

\begin{lemma}
\label{lemma:construct}
The constructed utility functions $\{ \overline{v}_a(r, \cdot) \}_{a,r}$ are piecewise-linear,  continuous, monotone decreasing and in these functions every slope is an integer power of $(1+\varepsilon$).
\end{lemma}
 
Using the above mentioned lemma we will now prove that the instance obtained by rounding the slopes approximately preserves envy freeness. 
\begin{lemma}[Approximation Guarantee]
\label{lemma:apx-guarantee}
Given an approximation parameter $\varepsilon \in (0,1)$ and a rent-division instance $\mathcal{I} = \left\langle \mathcal{A}, \mathcal{R}, \{ {v}_a(r, \cdot) \}_{a,r} \right\rangle$ with continuous, monotone decreasing, and piecewise-linear utility functions, we can efficiently construct an instance 
$\mathcal{\overline{I}} = \left\langle \mathcal{A}, \mathcal{R}, \{ \overline{v}_a(r, \cdot) \}_{a,r} \right\rangle$ such that (i) the utility functions $\{ \overline{v}_a(r, \cdot) \}_{a,r}$ are continuous, monotone decreasing, and piecewise linear, (ii) all the slopes in the utilities $\{ \overline{v}_a(r, \cdot) \}_{a,r}$ are integer powers of $(1+\varepsilon)$, and (iii) any envy-free solution of $\mathcal{\overline{I}}$ is an $\varepsilon$-$\EF$ solution of $\mathcal{I}$.
\end{lemma}

\begin{proof}
To obtain $\overline{\mathcal{I}}$ we construct the utility functions $\{ \overline{v}_a(r, \cdot) \}_{a, r}$ as mentioned above. Properties \emph{(i)} and \emph{(ii)} in the theorem statement follow directly from Lemma~\ref{lemma:construct}. It remains to prove property \emph{(iii)}. 

Let $(\pi,p)$ be an $\EF$ solution for $\mathcal{\overline{I}}$, i.e., for all agents $a$ and all rooms $r$ we have $\overline{v}_a(\pi(a),p_{\pi(a)}) \ge \overline{v}_a(r,p_r)$

First, consider all agents $a$ whose utilities are positive under the solution, i.e, consider $a$ such that $\overline{v}_a(\pi(a),p_{\pi(a)}) \geq 0 $ and let the price $p_{\pi(a)}$ be in the piece $ \left\langle [x_1, x_2], \lambda^a_{\pi(a)}, H^a_{\pi(a)}, B^a_{\pi(a)} \right\rangle$ of $v_a(\pi(a), \cdot)$. By \cref{lemma:same-sign}, we know that $v_a(\pi(a),p_{\pi(a)}) \ge 0$. We have the following chain of inequalities
\begin{align*}
v_a(r,p_r) & \leq \overline{v}_a(r,p_r) &	\text{(\cref{lemma:properties} (a))} \\
& \leq \overline{v}_a\left(\pi(a),p_{\pi(a)}\right) & \text{($(\pi, p)$ is $\EF$ for $\overline{\mathcal{I}}$)} \\
& \leq B^a_{\pi(a)} + \overline{\lambda}^a_{\pi(a)} \left(x_2 - p_{\pi(a)}\right) & \text{(\cref{lemma:properties} (b))} \\ 
& \leq B^a_{\pi(a)} + (1+\varepsilon) \ \lambda^a_{\pi(a)} \left(x_2 - p_{\pi(a)}\right) & \text{(using inequality (\ref{property:sandwich}))} \\
& \leq (1+\varepsilon)\ \left(B^a_{\pi(a)} + \lambda^a_{\pi(a)} (x_2 - p_{\pi(a)}) \right) & \text{(\cref{lemma:properties} (d))} \\
& = (1+\varepsilon) \ v_a\left(\pi(a),p_{\pi(a)}\right)  & \text{(using equality (\ref{eq:rewrite}))}
\end{align*}

Therefore, if $\overline{v}_a(\pi(a),p_{\pi(a)}) \geq 0 $, then $(\pi, p)$ is an $\varepsilon$-$\EF$ solution for the agent $a$. 
	
The complimentary case addresses agents $a$ for whom the utility is negative under the solution,i.e, $\overline{v}_a(\pi(a),p_{\pi(a)}) < 0 $. By \cref{lemma:same-sign}, we know that $v_a(\pi(a),p_{\pi(a)}) < 0 $.  Let $p_{\pi(a)}$ be in the piece $ \left \langle [x_1,x_2], \lambda^a_{\pi(a)}, H^a_\pi(a), B^a_\pi(a) \right \rangle$ of $v_a(\pi(a), \cdot)$ 

\begin{align*}
v_a(r,p_r) & \leq \overline{v}_a(r,p_r) &	\text{(\cref{lemma:properties} (a))} \\
&  \leq \overline{v}_a\left(\pi(a),p_{\pi(a)}\right) & \text{($(\pi, p)$ is $\EF$ for $\overline{\mathcal{I}}$)} \\
& \leq H^a_{\pi(a)} - \frac{\overline{\lambda}^a_{\pi(a)}}{1+\varepsilon} \left(x - x_1\right)  & \text{(\cref{lemma:properties} (c))} \\ 
& \leq H^a_{\pi(a)} - \frac{\lambda^a_{\pi(a)}}{1+\varepsilon} \left(x - x_1 \right)  & \text{(using inequality (\ref{property:sandwich}))} \\ 
& \leq \frac{1}{1+\varepsilon} \left( H^a_{\pi(a)} - {\lambda^a_{\pi(a)}} \left(x - x_1 \right) \right) & \text{(\cref{lemma:properties} (e))} \\
& = \frac{1}{1+\varepsilon}  \ v_a\left(\pi(a),p_{\pi(a)}\right) & \text{(using equality (\ref{eq:rewrite}))}
\end{align*}

Hence, for any agent $a$ with $v_a\left(\pi(a),p_{\pi(a)}\right)  <0$, the inequality $v_a\left(\pi(a),p_{\pi(a)}\right)  \geq (1+\varepsilon) v_a \left(r, p_r\right)$ holds for all rooms $r$. Overall, we get that $(\pi,p)$ is an $\varepsilon$-$\EF$ solution for $\mathcal{I}$.
	
\end{proof}

\cref{lemma:apx-guarantee} and \cref{theorem:ef-rounded} establish that the fair rent-division problem admits an FPTAS. 
 
\TheoremMainResult*

\section{Fair Rent Division with Fixed Total Rent}
\label{section:fixed}
This section considers the problem in which we are given a rent-division instance along with a total rent (fixed cost) $C$, and the objective is to find an $\EF$ solution in which the sum of prices is equal to $C$. We obtain an FPTAS for this problem by adapting $\ALG$ (the algorithm developed in \cref{section:rounded}). 

In contrast to the results obtained in Sections \ref{section:rounded} and~\ref{section:FPTAS}, the solutions obtained in this section might impose negative rents, i.e., require \emph{transfers} among agents. In fact, for a general value of $C$, it might be the case that envy freeness can be achieved only if one imposes negative rents on some rooms.\footnote{This case will arise if the total rent $C$ is strictly less than the sum of rents at the optimal price vector $p^*$; see ~\cref{section:def-opt-prices} for the definition of $p^*$.}  The algorithm developed in this section necessarily finds an approximately envy-free solution, but, efficiently determining whether a given instance admits a fair solution in which the prices are nonnegative and sum up to $C$ remains an interesting open question. 

Under this total-rent constraint, nonnegativity of utilities under fair solutions cannot be guaranteed either, e.g., consider the case in which the utility of each of the $n$ agents for every room $r$ is negative, if the price of the room $p_r \geq C/n$. However, under a natural scaling assumption, instances with quasilinear utilities are known to admit an $\EF$ solution in which the utilities are nonnegative~\cite{brams2001competitive}. Specifically, the assumption entails that for every agent $a$ the sum of base values of all the the rooms is at least $C$, i.e., $\sum_r v_a(r, 0) \geq C$. For quasilinear utilities, this translates to requiring that for every agent $a$ we have $\sum_{r} z^a_r \geq C$, where $z^a_r$ is the price at which agent $a$'s utility for room $r$ is equal zero. We will prove that if an input instance (now with piecewise-linear utilities) satisfies this assumption, then the solution computed by our algorithm endows nonnegative utilities to the agents. 

Given a rent-division instance $\mathcal{I}$ (wherein the utility functions are continuous, monotone decreasing, and piecewise linear) and a total rent $C$, first, we use \cref{lemma:apx-guarantee} to construct another instance $\overline{\mathcal{I}}$. By construction, in $\overline{\mathcal{I}}$ the slopes of all the utility functions are integer powers of $(1 + \varepsilon)$. 

We detail a modification of $\ALG$ below, which finds an $\EF$ solution, $(\pi, p)$, wherein the sum of prices is equal to $C$, i.e., $\sum_r p_r = C$. Since the utility function satisfy the powers-of-$(1+\varepsilon)$ property, the modified algorithm finds such an $\EF$ solution for $\overline{\mathcal{I}}$ in time polynomial in $1/\varepsilon$ and the input size. This runtime analysis is analogous to one given in the proof of \cref{theorem:ef-rounded}. 

Furthermore, the approximation guarantee of \cref{lemma:apx-guarantee} implies that $(\pi, p)$ is an $\varepsilon$-$\EF$ solution of the given instance $\mathcal{I}$. Hence, we obtain an FPTAS for computing a fair division of the total rent $C$. \\

\noindent
{\bf Modifying $\ALG$:} To obtain an $\EF$ solution of $\overline{\mathcal{I}}$ under the fixed-cost constraint, we will modify $\ALG$ as follows 
\begin{enumerate}
\item The threshold $M$ (used in equation (\ref{equation:bar-to-hat})) is set such that $\overline{v}_a(r, C) > \overline{v}_a(r', M)$ for all agents $a$ and rooms $r, r'$. 
\item In $\ALG$, the equality $\sum_r p_r^i = C$ is used as the termination condition of the while-loop, instead of $p_r = 0$, for some rooms $r$. 
\item In the linear programs LP$(L^{i-1}, U^{i-1}, \pi^i)$, replace the nonnegativity constraint (i.e., $x \geq \mathbb{0}$) with $\sum_r x_r \ge C$. 
\end{enumerate}

In the proof of the following theorem we will show that this version of $\ALG$ algorithm finds the desired envy-free solution.\footnote{As mentioned previosuly, in \cref{theorem:fixed-cost} the computed solution might impose negative prices (i.e., entail transfers) and lead to negative utilities for the agents.}  

\TheoremFixedCost*


\begin{proof}
Using \cref{lemma:apx-guarantee}, we construct a rent division instance $\overline{\mathcal{I}}$ such that all the slopes in this instance are integral powers of $(1+\varepsilon)$. The modified version of \cref{algorithm:envy-free} described above is executed with instance $\overline{\mathcal{I}}$ as the input. Let $\widehat{\mathcal{I}}$ denote the instance obtained in Step~\ref{step:hat-instance} of modified $\ALG$ and let $M$ be the threshold adapted for $C$. 

Note that \cref{lemma:ef-invariant}, \cref{lemma:wts-drop}, and \cref{lemma:price-drop-sale-sale} do not depend upon the nonnegativity constraint, they rely only on the envy-freeness and linear-domain constraints. Therefore, these lemmas are also application in the context of the modified algorithm. Also $\ell$, which denotes the total number of pieces across all the utility functions, continues to upper bound the number of linear domains traversed by the algorithm. Overall, the proof of correctness and runtime analysis of $\ALG$  directly hold for the modified algorithm as well. Hence, the modified algorithm finds an $\EF$ solution, $(\pi,p)$, of the instance $\widehat{\mathcal{I}}$ in time that is polynomial in $1/\varepsilon$ and the input size.


Furthermore, the changes ensure that $\sum_r p_r = C$. Note that there is at least one room $\rho$ with $p_\rho \leq C$, which in turn implies that $p_{r'} \leq M$ for all rooms $r'$: assume, for contradiction, that $p_r' > M$ and let room $r'$ be assigned to, say, agent $\alpha$, i.e. $\pi(\alpha) =r'$. With $p_\rho \le C$ we have

\begin{align*}
\widehat{v}_{\alpha}(\pi(\alpha), p_{\pi(\alpha)}) & < \widehat{v}_{\alpha}(\pi(\alpha), M)   \qquad &\text{(utility $\widehat{v}_{\alpha}(\pi(\alpha), \cdot)$ is monotone decreasing)} \\
& = \overline{v}_\alpha(\pi(\alpha), M) \qquad &\text{(follows from (\ref{equation:bar-to-hat}))} \\
& < \overline{v}_\alpha(\rho, C) \qquad &\text{(by the definition of $M$)} \\
& = \widehat{v}_\alpha(\rho, C) \qquad &\text{(follows from (\ref{equation:bar-to-hat}) and $M \ge C$)} \\
& \le \widehat{v}_\alpha(\rho, p_\rho) \qquad &\text{(utility $\widehat{v}_{\alpha}(\rho, \cdot)$ is monotone decreasing)} 
\end{align*}
This contradicts the fact that $(\pi,p)$ is an $\EF$ solution of $\widehat{\mathcal{I}}$. 

When all the prices are less than $M$,  the equality $\widehat{v}_a(r, p_r) = \overline{v}_a(r, p_r)$ holds for all $a$ and $r$. Thus, $(\pi,p)$ is an envy-free solution for $\overline{\mathcal{I}}$ and, invoking \cref{lemma:apx-guarantee}, we get that $(\pi,p)$ is an $\varepsilon$-$\EF$ solution of the original instance $\mathcal{I}$.
\end{proof}
 
Recall that $z^a_r$ denotes the price at which agent $a$'s utility for room $r$ is equal zero, i.e., $z^a_r$ satisfies $v_a(r, z^a_r) = 0$. 
We will complete the section by proving that nonnegative utilities can be ensured if $\sum_r z^a_r \geq C$ for all agents $a$.

\TheoremIr*

\begin{proof}
Executing the modified version of \cref{algorithm:envy-free} (described above) we obtain an allocation $\pi$ and a price vector $p \in \mathbb{R}^n$ such that $(\pi,p)$ is an $\varepsilon$-$\EF$ solution of $\mathcal{I}$ and $\sum_r p_r = C$. Here, the correctness and time complexity directly follow from  \cref{theorem:fixed-cost}.

Since $\sum_r z^a_r \geq C = \sum_r p_r$, for all agents $a$ there exists a room $r(a)$ such that $p_{r(a)} \leq z^a_{r(a)}$. If $v_a(\pi(a),p_{\pi(a)})  < 0$, then it must be the case that $v_a(r, p_r) < 0$ for all rooms $r$; recall that $(\pi, p)$ is an $\varepsilon$-$\EF$ solution of $\mathcal{I}$ (see \cref{definition:apx-ef}). However, this leads to a contradiction: $v_a (r(a), p_{r(a)}) \geq v_a( r(a), z^a_{r(a)}) =0$.  Therefore, we have $v_a(\pi(a),p_{\pi(a)})  \geq 0$ for all agents $a$, and this establishes the claim that the utilities of all the agents are nonnegative.  



\end{proof}

\section{DSIC Mechanism for Structured Instances}
\label{section:optimal-price}

This section presents an algorithm for computing optimal prices (as defined in \cref{section:def-opt-prices}).  In particular, given any rent-division instance, ${\mathcal{I}}$, the algorithm developed in this section (\cref{algorithm:opt-price}) finds an envy-free solution $(\pi^*, p^*)$ of ${\mathcal{I}}$ such that $p^* \in \mathbb{R}^n_+$ and for all other $\EF$ solutions $(\pi, p)$ with nonnegative prices, the following inequality holds: $p^* \leq p$. 

We will, in particular, bound the runtime of \cref{algorithm:opt-price} for input instances wherein the slopes of all the piecewise-linear utilities are integer powers of $(1+\varepsilon)$, with parameter $\varepsilon >0$. The developed algorithm finds the optimal price for any given instance, though we will show that if the utilities satisfy the powers-of-$(1+\varepsilon)$ property, then an optimal price can be found in time polynomial in $1/\varepsilon$ and the bit complexity of the input. 

We will prove that an $\EF$ solution is not optimal iff there exists a subset of rooms whose rents can be reduced while maintaining envy freeness and nonnegativity of prices. \cref{algorithm:opt-price} relies on this novel characterization of optimality and, in particular, efficiently identifies a subset of rooms (if one exists) whose prices can be reduced further. The characterization ensures that the algorithm must have found the optimal price, once no such subset exits. This method, in fact, provides an alternate, constructive proof of existence of optimal prices. 

As mentioned previously, for the rent-division problem with strategic agents, any algorithm that computes the optimal solution $(\pi^*, p^*)$ is guaranteed to be dominant strategy incentive compatible (DSIC)~\cite{sun2003general}. Hence, for settings in which the utility functions satisfy the powers-of-$(1+\varepsilon)$ property,  \cref{algorithm:opt-price} provides a DSIC mechanism with runtime polynomial in $1/\varepsilon$ and the input size.\footnote{Extending this result to obtain an efficient, DSIC mechanism for general fair rent-division instances is an interesting direction for future work.}

 
We start by proving  Lemma \ref{comp}, which provides useful comparative results between $\EF$ solutions. Next, in Lemma \ref{lemma:char} we detail the above mentioned characterization of optimal prices. 

The following notion of a \emph{tight set} corresponds to the edges in the first-choice graph which are not part of the current allocation. That is, for an envy-free solution $(\pi, p)$, the tight set contains agent-room pairs, $(a,r)$, where agent $a$ is not allocated room $r$, but it derives maximum possible utility from the room $r$ as well. 

\begin{definition}[Tight set]
Given an envy-free solution $(\pi, p)$ of a rent-division instance $\I = \left\langle \A, \R, \{ v_a(r , \cdot)\}_{a, r} \right\rangle$, define the tight set $\T_{(\pi, p)} \coloneqq  \{ (a,r) \in \A \times \R \mid   \pi(a) \neq r \text { and } v_a(\pi(a), p_{\pi(a)}) = v_a(r, p_r) \}$.
\end{definition}

  \begin{figure}[h]
 	\begin{center}
 		\includegraphics[scale=0.7]{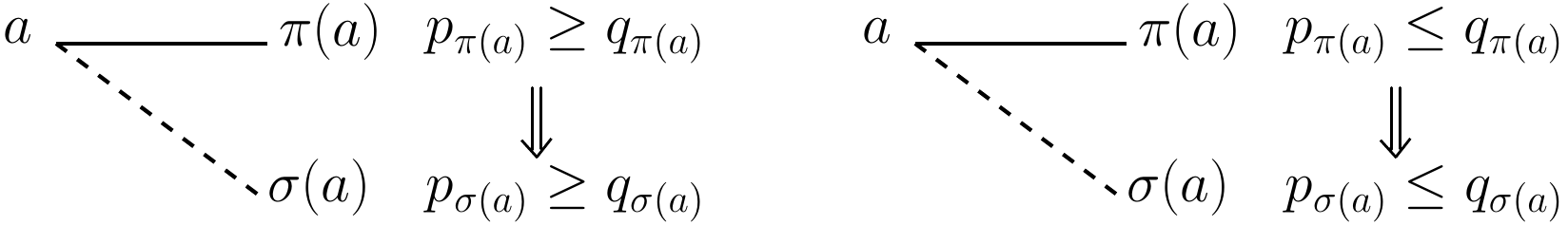}
 	\end{center}
 	\caption{Price comparisons between two envy-free solutions}
 	\label{fig:c12}
 \end{figure}

 \begin{figure}[h]
 	\begin{center}
 		\includegraphics[scale=0.7]{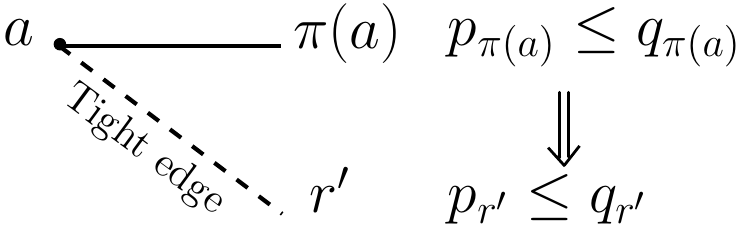}
 	\end{center}
 	\caption{Price comparison for tight edges}
 	\label{fig:c3}
 \end{figure}
 
\begin{lemma} \label{comp}
For any two envy-free solutions, $(\pi, p)$ and $(\sigma, q)$, of a rent-division instance $\mathcal{I}$, the following comparative results hold: 
	\begin{enumerate}[(a)]
		\item If for an agent $a \in \mathcal{A}$ we have $p_{\pi(a)} \geq q_{\pi(a)}$, then $p_{\sigma(a)} \geq q_{\sigma(a)}$.
		\item If for an agent $a \in \mathcal{A}$ we have $p_{\pi(a)} \leq q_{\pi(a)}$, then $p_{\sigma(a)} \leq q_{\sigma(a)}$.
		\item If for an agent $a \in \mathcal{A}$ we have $p_{\pi(a)} \leq q_{\pi(a)}$ and $(a, r') \in \mathcal{T}_{(\pi, p)}$, then $p_{r'} \leq q_{r'}$.
		\end{enumerate}		
\end{lemma}

\begin{proof}

We start by proving part $(a)$ of the Lemma. Assume, for contradiction, $p_{\sigma(a)} < q_{\sigma(a)}$, then following chain of inequalities contradicts the envy-freeness of $(\sigma, q)$.
\begin{align*}
v_a( \sigma(a), q_{\sigma(a)}) & < v_a (\sigma(a), p_{\sigma(a)}) & \text{(by assumption and monotonicity of $v_a(\sigma(a), \cdot)$)} \\
& \leq v_a (\pi(a), p_{\pi(a)}) & \text{(by envy-freeness of $(\pi, p)$)} \\
& \leq v_a (\pi(a), q_{\pi(a)}) & \text{(we are given that $p_{\pi(a)} \geq q_{\pi(a)}$)}
\end{align*}

For part $(b)$, let agent $a$ be contained in the orbit (cycle) $O=[a_0 = a, r_0, a_1, r_1, \ldots, a_{k-1}, r_{k-1}]$ obtained by composing the matchings $\sigma$ and $\pi$. Here, the rooms and agents are indexed (modulo $k$) such that, for all $i \in \{0,1, \ldots, k-1\}$, we have $r_i = \pi(a_i)$ and $r_i = \sigma(a_{(i+1)})$.

Note that the given condition ($p_{\pi(a)} \leq q_{\pi(a)}$) is equivalent to $q_{r_0} \geq p_{r_0}$, i.e., for agent $a_1$ we have $q_{\sigma(a_1)} \geq p_{\sigma(a_1)}$. Applying part $(a)$ with agent $a_1$ and matching $\sigma$ in hand we get  $q_{\pi(a_1)} \geq p_{\pi(a_1)}$, i.e., $q_{r_1} \geq p_{r_1}$. Since, $\sigma(a_2) = r_1$ we can inductively extend this argument to obtain $q_{r_{k-1}} \geq p_{r_{k-1}}$. That is, $q_{\sigma(a)} \geq p_{\sigma(a)}$, which establishes the claim (see Figure \ref{fig:c12} for part $(a)$ and $(b)$, respectively).

We complete the proof by establishing part $(c)$ of the lemma. Note that room $r'$ satisfies $v_a(\pi(a), p_{\pi(a)}) = v_a(r', p_{r'})$. Assume, for contradiction, that $p_{r'} > q_{r'}$. The inequalities below contradict the envy-freeness of $(\sigma, q)$ and, hence, complete the proof (see Figure \ref{fig:c3}).

\begin{align*}
v_a(\sigma(a), q_{\sigma(a)}) & \leq v_a(\sigma(a), p_{\sigma(a)}) & \text{(using part $(b)$)} \\
& \leq v_a(\pi(a), p_{\pi(a)}) & \text{(by envy-freeness of $(\pi, p)$)} \\
& = v_a(r', p_{r'}) & \text{(since $(a, r') \in \mathcal{T}_{(\pi, p)}$)} \\
& < v_a(r', q_{r'}) & \text{(by assumption)}  
\end{align*}	
\end{proof}

The following constructs will be used in the computation of optimal prices.

 \begin{definition}
Let $(\pi, p)$ be an envy-free solution of a given rent-division instance $\I$. Define directed graph $\W_{(\pi, p)} := (\A \cup \R, \overleftarrow{\pi} \cup \overrightarrow{\T_{(\pi, p)}})$, i.e., the directed edge set of $\W_{(\pi, p)}$ is equal to $\{(r, a) \in \mathcal{R} \times \mathcal{A} \mid  \pi(a) =r \} \cup \{ (a',r') \in \mathcal{A} \times \mathcal{R} \mid (a',r') \in \T_{(\pi, p)} \}$.
\end{definition}

\begin{definition}
\label{definition:ep}
Let $\mathcal{I}$ be a rent division instance with optimal price $p^*$. For any envy-free solution of $\mathcal{I}$, say $(\pi, p)$, with nonnegative prices, define the following sets:
\begin{itemize}
\item $E_p := \{r \in \mathcal{R} \mid p_r \leq p^*_r\}$. 
\item $E_p^c := \R \setminus E_p$
\item $Z_p := \{r \mid p_r = 0\}$. Note that $Z_p \subseteq E_p$.
\end{itemize}
\end{definition}

{
	\begin{algorithm}[h]
		{
			{\bf Input:} A rent-division instance ${\mathcal{I}} = \left\langle \mathcal{A}, \mathcal{R}, \{ {v}_a(r, \cdot) \}_{a,r} \right\rangle$  with continuous, monotone decreasing, piecewise-linear utility functions \\
			{\bf Output:} The optimal envy-free price of ${\mathcal{I}}$
			\caption{Algorithm for computing optimal envy-free prices}
			\label{algorithm:opt-price}
			\begin{algorithmic}[1]
				\STATE Set $(\pi, p)$ to be the output of $\ALG$ on $\mathcal{I}$ \\
				\COMMENT{Note that in $(\pi,p)$ there exists at least one room $r$ such that $p_r=0$}
				\STATE Compute $E_p$ and $E_p^c$ using the characterization given in \cref{lemma:char}
				\WHILE{$E_p^c \neq \emptyset $}
				\STATE Set $\mathcal{I'} = \left\langle \mathcal{A'}, \mathcal{R'}, \{ {v}_a(r, \cdot) \}_{a,r} \right\rangle$, with $\mathcal{A'} = \pi^{-1}(E_p^c), \mathcal{R'} = E_p^c$ 
				\STATE Run $\ALG$ on the sub-instance $\mathcal{I'} $ with initial solution set as the restriction of $(\pi, p)$ to $\mathcal{I'}$. Also, modify $\ALG$ to include the following linear constraints in each LP instantiation  \\ $v_a(\pi(a), p_{\pi(a)}) \geq v_a(r, p_r)$  for all $ a \in \pi^{-1}(E_p)$ and $ r \in E_p^c$. 
				\STATE Ensure that $\ALG$ terminates if $ v_a(\pi(a), p_{\pi(a)}) = v_a(r, p_r)$ for any agent $a \in \pi^{-1}(E_p)$ and room $r \in E_p^c$ \COMMENT{We stop $\ALG$ if a new incoming tight edge is formed from an agent in $\pi^{-1}(E_p)$ to a room in $E_p^c$}
				\STATE Set $(\pi', p')$ to be the solution returned by $\ALG$ as executed in the previous two steps
				\STATE Update $\pi$ and $p$ as follows 
				\begin{align*}
				\pi & = 
				\begin{cases}
				 \pi(a) & \text{ for }  a \in \pi^{-1}(E_p) \\
				 \pi'(a) & \text { for } a \in \pi^{-1}(E^c_p) 
				\end{cases} \\
				p & =
				\begin{cases}
				p_r & \ \ \ \ \ \text{  for }  r \in E_p \\
				p'_r & \ \ \ \ \ \text {  for } r \in E^c_p
				\end{cases}
				\end{align*}
				\STATE Update $E_p$ and $E_p^c$
				\ENDWHILE
				
				\STATE Return $(\pi, p)$
			\end{algorithmic}
		}
	\end{algorithm}
}

Note that  the optimality of $p^*$ ensures that all the inequalities in the definition of $E_p$ are tight for any nonnegative price vector $p$. In addition, a price vector $p$ is optimal iff the set $E_p^c = \emptyset$. 

Since the definition of $E_p$ directly refers to the optimal price vector, $p^*$ (which is not known a priori), it is unclear if one can efficiently find this set for a given price vector $p$. By contrast, the subset $Z_p$ can be directly computed.  The following lemma shows that we can, in fact, efficiently find $E_p$, even if the optimal price vector $p^*$ is not known.  


\begin{lemma} [Characterization of $E_p$] 
\label{lemma:char}
For any envy-free solution $(\pi, p)$ with nonnegative prices, a room $r$ belongs to the set $E_p$ iff there exists a room $z \in Z_p$ such that $r$ is reachable from $z$ in the directed graph $\W_{(\pi, p)}$. 
\end{lemma}

\begin{proof}
For the given instance, write $(\pi^*, p^*)$ to denote an $\EF$ solution in which $p^*$ is the optimal price vector. \\

\noindent
\emph{If part:} For room $z \in Z_p$ we have $p_z = p^*_z $. Therefore, a repeated application of \cref{comp} (c) (staring with $a = \pi^{-1}(z)$ and $(\sigma, q) = (\pi^*, p^*)$) establishes the claim.  \\

\noindent 
\emph{Only if part:} Consider a room $r$ which is not reachable from any $z \in Z_p$. We will prove that in such a case $r \notin E_p$.


Write $S(r)$ to denote the set of rooms that have a directed path to $r$ in $\W_{(\pi, p)}$, i.e., $S(r) \coloneqq \{ r' \in \mathcal{R} \mid r \text{ is reachable from } r' \text{ in } \W_{(\pi, p)} \}$. Note that, ${S}(r) \cap Z_p = \emptyset$, since $r$ is not reachable from $Z_p$. In other words, $p_{r'} >0 $ for all $r' \in S(r)$. Also, let $\pi^{-1}({S}(r))$ denote the set of agents that are allocated a room from the set $S(r)$, i.e., $\pi^{-1}({S}(r)) \coloneqq \{ a \in \mathcal{A} \mid \pi(a) \in S(r) \}$.

The definition of $S(r)$ ensures that there are no incoming tight edges in this set from an agent in $\mathcal{A} \setminus \pi^{-1}({S}(r)) $, i.e., there does not exist $(a',r') \in \T_{(\pi, p)} $ such that  $a' \in \mathcal{A} \setminus \pi^{-1}({S}(r)) $ and $r' \in S(r)$: the existence of such an $(a',r')$ would imply that the room $\pi(a')$ has a directed path to $r$ (through $r'$) and, hence, $\pi(a')$ would have been included in $S(r)$ and $a'$ in $\pi^{-1}({S}(r))$.



Therefore, for every room $r' \in \mathcal{S}(r)$, we have (i) $p_{r'} >0$ and (ii) there does not exist a tight edge $(a',r') \in \T_{(\pi, p)} $ such that $a' \notin \pi^{-1}({S}(r))$.

Now, consider the ``sub-instance'' $\mathcal{I'} \coloneqq \left\langle \mathcal{A'}, \mathcal{R'}, \left\{ {v}_a(r, \cdot) \right\}_{a \in \mathcal{A'}, r \in \mathcal{R'}} \right\rangle$, with $\mathcal{A'} = \pi^{-1}({S}(r)), \mathcal{R'} = {S}(r)$. Note that $(\pi, p)$ restricted to $\mathcal{I'}$ is an envy-free solution for the sub-instance. Hence, we can apply the perturbation lemma (\cref{lemma:original-perturb}) to strictly reduce the prices of all rooms in the set $\mathcal{S}(r)$, till either price of a room becomes zero or an incoming tight edge to the set ${S}(r)$ is formed.  Overall, this implies that the price of room $r$ can be decreased while maintaining envy-freeness. Hence, it must be that case that $p_r > p^*_r$, i.e., $r \notin E_p$ and the claim follows.
\end{proof}

This lemma implies that, for any given envy-free solution $(\pi, p)$, we can efficiently find the sets $E_p$ and $E^c_p$ by, say, a breadth-first search traversal of $\W_{(\pi, p)}$. We next state and prove the main result of this section.

   

\TheoremPostpro* 

\begin{proof}
 For analysis, write $(\pi, p)$ to denote the envy-free allocation at the beginning of an iteration of the while-loop in \cref{algorithm:opt-price}, and let $E_p$ and $E_p^c$ be the partition of rooms induced by $(\pi, p)$. Using the set $E_p^c$, the algorithm constructs the sub-instance $\mathcal{I'} = \left\langle \mathcal{A'}, \mathcal{R'}, \{ \overline{v}_a(r, \cdot) \}_{a, r} \right\rangle$ (recall that $\mathcal{A'} = \pi^{-1}(E_p^c)$ and $\mathcal{R'} = E_p^c$). Write $(\pi', p')$ to denote the envy-free solution computed for $\mathcal{I'}$. We will use $(\tilde{\pi}, \tilde{p})$ to denote the updated solution obtained at the end of this iteration
 
   \begin{align*}
  \tilde{\pi}(a) & = 
  \begin{cases}
  \pi(a) & \text{ for }  a \in \pi^{-1}(E_p) \\
  \pi'(a) & \text { for } a \in \pi^{-1}(E^c_p) 
  \end{cases} \\
  \tilde{p}_r & =
  \begin{cases}
  p_r & \ \ \ \ \ \text{  for }  r \in E_p \\
  p'_r & \ \ \ \ \ \text {  for } r \in E^c_p
  \end{cases}
  \end{align*}

Note that during the iteration, the prices of the rooms in $E_p$ do not change, while the prices of the rooms in $E_p^c$ decrease, i.e., $ \tilde{p}_r = p_r $ for all $ r \in E_p $ and $\tilde{p}_r < p_r$ for all $r \in E_p^c$.

The case analysis below shows that envy-freeness is maintained as an invariant, i.e., $(\tilde{\pi}, \tilde{p})$ is an $\EF$ solution of $\overline{\mathcal{I}}$. \\

\noindent
\emph{Case {\rm I}:} Agent $a \in \pi^{-1}(E_p)$ and room $ r \in E_p$ 
	\begin{align*}
	v_a(r,\tilde{p}_r) & = v_a(r, {p}_r) &  \text{(since $\tilde{p}_r = p_r$)} \\
	& \leq v_a(\pi(a),p_{\pi(a)})  & \text{(by envy-freeness of $(\pi,p)$)} \\
	& = v_a(\tilde{\pi}(a),\tilde{p}_{\tilde{\pi}(a)}) & \text{(since $\tilde{\pi}(a) = \pi(a)$ and $\tilde{p}_{\tilde{\pi}(a)} = p_{\pi(a)}$)}
	\end{align*}	

\noindent
\emph{Case {\rm II}:} Agent $a \in \pi^{-1}(E_p)$ and room $ r \in E_p^c$
	\begin{align*}
	v_a(r,\tilde{p}_r)  & \leq v_a(\pi(a),p_{{\pi}(a)}) & \text{(constraints imposed in Step 5 of \cref{algorithm:opt-price})} \\
	& =  v_a(\tilde{\pi}(a),\tilde{p}_{\tilde{\pi}(a)})  & \text{(since $\tilde{\pi}(a) = \pi(a)$ and $\tilde{p}_{\tilde{\pi}(a)} = p_{\pi(a)}$)} 
	\end{align*}

\noindent
\emph{Case {\rm III}:} Agent $a \in \pi^{-1}(E_p^c)$ and room $ r \in E_p$
    
     \begin{align*}
     v_a(r,\tilde{p}_r) & = v_a(r, {p}_r) &  \text{(since $\tilde{p}_r = p_r$)} \\
     & \leq v_a(\pi(a),p_{\pi(a)}) & \text{(by envy-freeness of $(\pi,p)$)} \\
	& \leq v_a(\pi(a), p'_{\pi(a)}) & \text{(since $p'_{\pi(a)} < p_{\pi(a)}$)} \\
	& \leq v_a(\pi'(a),p'_{\pi'(a)}) & \text{(by envy-freeness of $(\pi',p')$)} \\
	& =  v_a(\tilde{\pi}(a),\tilde{p}_{\tilde{\pi}(a)})  & \text{(since $\tilde{p}_{\tilde{\pi}(a)} = p'_{\pi'(a)}$ and $\tilde{\pi}(a) = \pi'(a)$)}
     \end{align*}

\emph{Case {\rm IV}:} Agent $a \in \pi^{-1}(E_p^c)$ and room $ r \in E_p^c$
    
    \begin{align*}
    v_a(r, \tilde{p}_r)  & = v_a(r, p'_r) &  \text{(since $\tilde{p}_r = p'_r$)} \\
    & \leq v_a(\pi'(a),p'_{\pi'(a)}) & \text{(by envy-freeness of $(\pi',p')$)} \\
 & = v_a(\tilde{\pi}(a),\tilde{p}_{\tilde{\pi}(a)})  & \text{(since $\tilde{\pi}(a) = \pi'(a)$ and $\tilde{p}_{\tilde{\pi}(a)} = p'_{\tilde{\pi}(a)}$)}
    \end{align*}
Therefore, $v_a(\tilde{\pi}(a), \tilde{p}_{\tilde{\pi}(a)}) \geq v_a(r, \tilde{p}_r)$ for all agents $a \in \A,$ and rooms $r \in \R$. In other words, envy-freeness is maintained throughout the execution the algorithm. 

Note that in every iteration of \cref{algorithm:opt-price} either the rent of a room in $E_p^c$ reduces to zero or a new tight edge is created, which, in turn, moves a room into the set $E_p$. Therefore, in every iteration, the cardinality of the set $E_p^c$ decreases. Since the algorithm terminates when $E_p^c = \emptyset$, the while-loop in the algorithm executes at most $n$ times. This observation and the time complexity of \cref{algorithm:envy-free} (which is used as a subroutine in \cref{algorithm:opt-price}) establishes the stated runtime bound.

Finally, using the fact that $p$ is the optimal price vector iff $E_p^c = \emptyset$ (\cref{definition:ep}) we get that the price vector  returned by the algorithm is indeed optimal.  
\end{proof}


\section{Complexity of Fair Rent Division}
\label{section:complexity}

The section studies the complexity of finding $\EF$ solutions of rent-division instances with continuous, monotone decreasing, and piecewise-linear utilities. Specifically, we show that this total problem $\plc$ (formally defined below) is contained in the intersection of the complexity classes $\PPAD$ and $\PLS$.

An instance of the the problem $\plc$ comprises of a tuple $ \langle \A, \R, \{v_a(r, \cdot) \}_{a, r }  \rangle$, where  $\mathcal{A} =[n]$ denotes the set of $n$ agents, $\mathcal{R} =[n]$ denotes the set of $n$ rooms, and $v_a(r, \cdot)$ specifies the utility of agent $a$ for room $r$. Here, each function $v_a(r, \cdot)$ is continuous, monotone decreasing, and piecewise linear in the price. In particular, $v_a(r, \cdot)$ is given using its base value $v_a(r,0)$, a set of increasing {break points}, $b_1 = 0, b_2, b_3, \ldots, b_t \in \mathbbm{R}_+$, and the magnitude of the slopes $\{ \lambda^a_{r,i} \in \mathbb{R}_+ \}_{i \in [t]}$. The objective of the problem is to find an envy free solution $(\pi, p)$, i.e., a solution that satisfies
$ v_a(\pi(a), p_{\pi(a)}) \geq v_a(r, p_r) $  for all $a \in \A$ and $r \in \R$.

The main result of this section is as follows.

\TheoremComplexity*

The proof of containment of $\plc$ in $\PPAD$ and $\PLS$ appears in \cref{subsection:ppad} and \cref{subsection:pls}, respectively. To show that the problem belongs to the complexity class $\PPAD$, we first introduce a total problem $\linear$, which entails finding $\EF$ solutions under linear utilities. Then, we reduce $\linear$ to the problem of finding a Nash equilibrium of a polymatrix game and complete the proof by reducing $\plc$ to $\linear$.

The $\PLS$ containment of $\plc$ is based on a potential argument which follows from the analysis of \cref{algorithm:envy-free}. 

\subsection{Proof of PPAD containment}                        
\label{subsection:ppad}
The complexity class $\PPAD$ (Polynomial Parity Arguments on Directed graphs) consists of total search problems that reduce to the canonical, complete problem $\textsc{EndOfTheLine}$. In $\textsc{EndOfTheLine}$, we are implicitly given an exponential-size directed graph, which consists of paths and cycles, along with a vertex that has no incoming edge. The objective is to find another vertex with degree one; the existence of such a vertex follows from a parity argument. 

We will first consider rent division with linear utilities and show that this total problem, \linear, reduces to the $\PPAD$-complete problem of finding an exact Nash equilibrium in polymatrix games. Furthermore, we will reduce $\plc$ to $\linear$ (\cref{plc-to-linear}). Hence, this chain of reductions establishes the stated \PPAD containment.

An instance of the total problem $\linear$ comprises of a tuple $ \langle \A, \R, \{v_a(r, \cdot) \}_{a, r}  \rangle$, where  $\mathcal{A} =[n]$ denotes the set of $n$ agents, $\mathcal{R} =[n]$ denotes the set of $n$ rooms, and $v_a(r, \cdot)$ specifies the utility of agent $a$ for room $r$. Here, each function $v_a(r, \cdot)$ is linear and monotone decreasing in the price. In particular, for all $x \in \mathbb{R}$, we have $v_a(r, x) := H^a_r - \lambda^a_r x$; where $H^a_r$ is the given base value and $\lambda^a_r$ is the magnitude of the slope. The objective of the problem is to find an envy free solution $(\pi, p)$ with a non-negative price vector, i.e., a solution that satisfies $v_a(\pi(a), p_{\pi(a)})  \geq v_a(r, p_r) $   for all $ a \in \A$ and $ r \in \R$ along with $ p_r \geq 0$ for all $r \in \R$.

\subsubsection{Reduction}
We prove that $\linear$ belongs to $\PPAD$ by reducing it to the problem of computing an exact Nash equilibrium in polymatrix games. Given that equilibrium computation in such games is known to be $\PPAD$-complete~\cite{daskalakis2006game}, we get that $\linear \in \PPAD$. Intuitively, we construct a polymatrix game in which each agent plays a bimatrix game against a ``landlord''. These bimatrix games are constructed such that expected payoffs of the players correspond to utilities of the corresponding agents, under a transformation from the mixed strategies to the prices. Furthermore,  the mixed strategy of the landlord, at an equilibrium, maps to an envy-free price vector. Hence, we obtain a reduction from rent division to equilibrium computation. 

A polymatrix game is a game where the interactions between the players are succinctly captured by a graph. The players correspond to nodes of the graph and every edge $(a, b)$ in the graph captures a bimatrix game between agents $a$ and $b$. Every player plays a single (mixed) strategy and her total payoff is sum of the payoffs from all the bimatrix games she plays with her neighbors in the graph. 

Given an $\linear$ instance $\mL = \langle \A, \R, \{v_a(r, \cdot) \}_{a, r}  \rangle$, we will construct a polymatrix game $\G$ with $n+2$ players and $n+1$ actions per player. Also, let $M \in \mathbbm{R}_+$ be a large threshold such that for all agents $a \in [n]$ and for all rooms $r, r' \in [n]$, we have $v_a(r, M) < v_a(r', 0)$. 

In the constructed game $\G$ players $1$ to $n$ correspond to the $n$ agents $\A$ of the given $\linear$ instance. As a gadget, we introduce player $(n+1)$ and room/action $(n+1)$. This additional player and action will be crucial in establishing certain properties for Nash equilibria of the polymatrix game. Player $0$ in $\G$ represents the ``landlord'' and, at an equilibrium, the mixed strategy of this player maps to an envy-free price vector. 

The utility functions of the gadget agent $(n+1)$ are defined as follows
\begin{align*}
v_{n+1}(r, z) := \begin{cases}
\gamma - \frac{\gamma}{2M} z & \text{for rooms $r \in [n]$}\\
\gamma & \text{for room $r = n+1$} 
\end{cases}
\end{align*}


The utility of the extra room is set to be a constant for each agent $a \in [n]$; in particular, this utility is defined in the following manner, $v_a(n+1, z) \coloneqq  \max_{r \in [n]} v_a(r,M) + \eta$,  
where $\eta$ is a small positive number such that $\max_{r \in [n]} v_a(r,M) + \eta < \min_{r' \in [n]} v_a(r',0)$. The definition of $M$ ensures that such an $\eta$ exists. Note that all the utilities in this extended rent-division instance continue to be linear functions. 

Write $\{0,1,...,n+1\}$ to denote the set of  $(n+2)$ players in the polymatrix game $\G$. Here, player $1$ to player $(n+1)$ correspond to the $(n+1)$ agents of the extended rent division instance while player $0$ corresponds to the landlord. The action set of every player in the game is the set of $(n+1)$ rooms, denoted by $[n+1]$. 


The polymatrix game $\G$ consists of $(n+1)$ bimatrix games. Specifically, player $0$ plays a bimatrix game $(P^a, Q^a)$ with every other player $a \in \{1,2,...,n+1\}$. Here, player $a$ is the row player and player $0$ is the column player with $(n+1) \times (n+1)$ payoff matrices $P^a$ and  $Q^a$, respectively. Since the utility function for each agent $a \in [n+1]$ and each room $r \in [n+1]$ is a linear function of the form $v_a(r, z) = H^a_r - \lambda^a_r z$, we can construct payoff matrices wherein the expected payoff maps to the utility of the room. Formally, define payoff matrix $P^a$ for every player $a \in [n+1]$ as 

\begin{align*}
P^a_{(r,r')} & := 
\begin{cases}
H^a_r - \lambda^a_r M & \text{ for }  r \neq r' \\
(H^a_r - \lambda^a_rM) + 3n\lambda^a_rM & \text { for }  r=r'
\end{cases} 
\end{align*} 

For all players $a$ in $[n]$, the payoff matrix $Q^a$ of player $0$ is set to be the negative identity matrix that is, $Q^a := - I$. The payoff matrix of player $0$ for the game with player $(n+1)$ is set to be the following diagonal matrix 

\begin{align*}
Q^{n+1}_{(r,r')} & := 
\begin{cases}
0 & \text{ for }  r \neq r' \\
-(n+1)^2 & \text { for }  r=r' \neq (n+1) \\
-1 & \text { for }  r=r'=(n+1)
\end{cases} 
\end{align*} 

Write $x^a \in \Delta^n$ to be a mixed strategy of Player $a$ over the action set of $(n+1)$ rooms. Hence, $(x^0,x^1,...,x^{n+1})$ denotes a strategy profile for the above constructed polymatrix game. The construction ensures that a mixed strategy $x^0$ of the (column) player $0$ induces a price vector $p$ for the underlying rent-division instance. In particular, the expected payoff, $P^a_r(x^0)$, of a row player $a \in [n+1]$ for an action $r \in [n+1]$, when the column player $0$ plays mixed strategy $x^0$ is $P^a_r(x^0) = v_a(r, M(1 - 3nx^0_r))$. 
In other words, the expected payoff for action $r \in [n+1]$ is equal to the utility at price $p_r := M(1-3n \ x^0_r)$. Note that if player $0$ plays action $r$ with zero probability, i.e., $x^0_r = 0$, then this is equivalent to setting the price of the room $p_r$ to be equal to $M$. However, if player $0$ plays action $r$ with probability $1/3n$, then this leads to zero rent, $p_r =0$. 

Write $BR_a(x^0)$ to denote the best-response set of player $a$ against strategy $x^0$ of the column player, $BR_a(x^0) \coloneqq  \{r \in [n+1] \mid P^a_r(x^0) \geq P^a_{r'}(x^0) \ \ \text{for all} \ \ r' \in [n+1] \}$. Note that the expected payoffs of player $(n+1)$ against the mixed strategy $x^0 \in \Delta^n$ of player $0$ are as follows: $P^{n+1}_{n+1}(x^0) = \gamma$ and  $P^{n+1}_r(x^0) = \frac{\gamma}{2} + \frac{3\gamma n}{2}x^0_r $ for $r \in [n]$. Therefore, we can identify the conditions under which an action $r \in [n]$ is a best response for player $(n+1)$: 

\begin{align}
\label{threshold}
\begin{cases}
\text{If} \ x^0_r < 1/3n,  & \text{then room} \ r \ \text{is not a best response for player} \ (n+1)\\
\text{If} \ x^0_r = 1/3n, & \text{then player} \ (n+1) \ \text{is indifferent between room} \ r \ \text{and room} \ (n+1)\\
\text{If} \ x^0_r > 1/3n,  & \text{then room} \ (n+1) \ \text{is not a best response for player} \ (n+1)\\
\end{cases}
\end{align}

The expected payoff of player $0$ for an action $r \in [n+1]$, when the remaining players play mixed strategies $(x^1, x^2,...,x^{n+1})$, admits a closed form:
\begin{align}
\label{colpayoff}
Q_r(x^1, x^2,...,x^{n+1}) & = 
\begin{cases}
-\left(\sum_{a \in [n]} x^a_r +(n+1)^2x^{n+1}_r\right) & \text{ for }  r \in [n] \\
-\left(\sum_{a \in [n+1]} x^a_r\right) & \text { for }  r= n+1 
\end{cases} 
\end{align} 

We can write the best response set of player $0$ as $BR_0(x^1, x^2,...,x^{n+1}) = \{r \in [n+1] \mid Q_r(x^1, x^2,...,x^{n+1}) \geq Q_{r'}(x^1, x^2,...,x^{n+1})  \ \text{for all} \ \ r' \in [n+1] \}$

Recall that, under any Nash equilibrium, the support of the mixed strategy of any player is contained in her best response set. Therefore, the following containments hold for any Nash equilibrium $(x^0, x^1,...,x^{n+1})$ of the constructed polymatrix game, $ \text{Supp}(x^0) \subseteq BR_0(x^1, x^2,...,x^{n+1})$ and $\text{Supp}(x^a) \subseteq BR_a(x^0)$ for all $a \in [n+1]$.

The key component of the reduction is the following lemma, the proof of which provides an efficient method of converting a Nash equilibrium of the constructed polymatrix game to an envy-free solution for the given rent-division instance.  

\begin{lemma} \label{ppadreduction}
	Given any Nash equilibrium $(x^0, x^1,...,x^{n+1})$ of the constructed polymatrix game $\G$, we can efficiently find an envy-free solution $(\pi, p)$ of the underlying $\linear$ instance $\mL$. 
\end{lemma}

\subsubsection{Proof of \cref{ppadreduction}}
We will first establish relevant properties that any Nash equilibrium of the game $\G$ must satisfy. Here, the end goal is to show that an $n \times n$ matrix formed by stacking the equilibrium mixed strategies $(x^1, x^2, \ldots ,x^n)$ as the rows is doubly stochastic. Therefore, the mixed strategies can be associated with a permutation $\pi$. Furthermore, the equilibrium mixed strategy $x^0$ of the landlord (player $0$) leads to a price vector $p$ for the $n$ rooms. 
The following five claims establish these facts and will enable us to prove that $(\pi, p)$ is an $\EF$ solution of the underlying rent division instance. 


The first claim asserts that in any Nash equilibrium the $(n+1)$th player selects room $(n+1)$ with nonzero probability. 

\begin{claim} \label{Claim1}
	Under any Nash equilibrium $(x^0, x^1, x^2, \ldots ,x^n)$ of the constructed game $\G$, player $(n+1)$ plays the action $(n+1)$ with nonzero probability, $x_{n+1}^{n+1} > 0$.
\end{claim}

\begin{proof}
	Assume, for contradiction, that $x_{n+1}^{n+1} = 0$. Since $x^{n+1} \in \Delta^n$, there exists an action $\rho \in [n]$ such that $x^{n+1}_\rho \geq 1/n$. Note that, at any Nash equilibrium the support of the mixed strategy for a player is contained in her best response set. Hence, action $\rho \in BR_{n+1}(x^0)$. Conditions (\ref{threshold}) ensure that this can happen only if player $0$ plays action $\rho$ with a probability at least $1/3n$, i.e., $x^0_\rho \geq 1/3n >0$. In addition, the expected payoff of player $0$ for the action $\rho$, under the equilibrium $(x^0, x^1,...,x^{n+1})$, is equal to  
	$Q_{\rho}(x^1, x^2,...,x^{n+1}) = -(\sum_{a \in [n]} x^a_\rho +(n+1)^2x^{n+1}_\rho) < - (1/n)(n+1)^2 < -(n+1)$.
	
	Independent of $(x^1, x^2,\ldots, x^{n+1})$, if column player $0$ plays action $(n+1)$, then her expected payoff will be at least $-(n+1)$; in particular, $Q_{n+1}(x^1, x^2,...,x^{n+1}) > -(n+1)$. Hence, player $0$ achieves a strictly greater payoff from action $(n+1)$ than action $\rho$, i.e., $\rho \notin BR_0(x^1, x^2,...,x^{n+1})$. This contradicts the fact that $x^0_\rho >0$. 
\end{proof}

The next claim shows that the equilibrium strategy of player $0$ is always full support.

\begin{claim} \label{fullsupport}
	Under any Nash equilibrium $(x^0, x^1, x^2, \ldots ,x^n)$ of the polymatrix game $\G$, the mixed strategy of player $0$ is full support, i.e., $x^0_r > 0$ for all $r \in [n+1]$. 
\end{claim}

\begin{proof}
	We will prove this claim in two parts. First, we will show that player $0$ places nonzero probability mass on the $(n+1)$th action and then use this fact to show that $x^0_r >0 $ for all $r$. \\
	
	\noindent 
	\textit{Part 1:}  Assume, for contradiction, that $x_{n+1}^0 =0$. Since, $x^0 \in \Delta^n$, there exists an action $r \in [n]$ such that $x^0_r \geq 1/n >1/3n $. Condition (\ref{threshold}) then ensure that action $(n+1)$ is not a best response for player $(n+1)$, that is $x^{n+1}_{n+1} =0$. This contradicts \cref{Claim1}. Hence, we must have $x_{n+1}^0 > 0$. \\
	
	\noindent 
	\textit{Part 2:} Assume, for contradiction, that $x^0$ is not full support. Hence, there exists a room $\rho \in [n]$ such that $x^0_\rho=0$. This implies that  the corresponding price of this room $p_\rho = M(1 - 3n \ x^0_\rho) = M$ and, hence, this room is not a best response for any agent, i.e.,  $\rho \notin BR_a(x^0)$ for all $a \in [n+1]$; recall that, for any agent $a \in [n]$ and mixed strategy $x^0$, the expected payoff of action $(n+1)$ is equal to $\max_{r\in [n]} v_a(r, M) + \eta$, which is strictly greater than $v_a(\rho, M)$. 
	
	In other words, $x^0_\rho =0$ implies $x^a_\rho =0$ for all $a \in [n+1]$. This leads to player $0$ receiving a payoff of zero at action $\rho$, i.e.,  $Q_\rho(x^1, x^2,...,x^{n+1}) =0$. 
	
	\cref{Claim1} ensures that  under any Nash equilibrium of the game we have $x_{n+1}^{n+1}>0$. This inequality along with \cref{colpayoff} imply that player $0$ receives a strictly negative payoff at action $(n+1)$, i.e., $Q_{n+1}(x^1, x^2,...,x^{n+1}) <0$.  \textit{Part 1} of this claim  has already established that player $0$ plays action $(n+1)$ with a nonzero probability at any Nash equilibrium, i.e., action $(n+1)$ is a best response of player $0$ under any Nash equilibrium. However, the inequality $Q_\rho(x^1, x^2,...,x^{n+1}) > Q_{n+1}(x^1, x^2,...,x^{n+1})$ contradicts this fact and, hence, establishes the claim. 
\end{proof}

Using the result that player $0$ plays every action with nonzero probability, we will establish an upper bound on player $0$'s payoff for each room.

\begin{claim}
	\label{claim:claim3}
	Under any Nash equilibrium $(x^0, x^1, \ldots ,x^n)$ of the game $\G$ and for every pair of actions $r, r' \in [n+1]$ we have $Q_r(x^1,x^2,...,x^{n+1}) =Q_{r'}(x^1,x^2,...,x^{n+1}) \le -1$.
\end{claim}
\begin{proof}
	\cref{fullsupport} implies that Supp$(x^0) = [n+1]$. Therefore, under a Nash equilibrium, every action is a best response for player $0$. In other words, this  player achieves the same expected payoff from all possible actions, $Q_r(x^1,x^2,...,x^{n+1}) =Q_{r'}(x^1,x^2,...,x^{n+1})$ for all $r, r' \in [n+1]$. 
	
	
	\cref{colpayoff} gives us the following upper bound on the expected payoff of player $0$ for each room $r \in [n+1]$: $Q_r(x^1,x^2,...,x^{n+1}) \leq - \sum_{a \in [n+1]} x^a_r$. Adding up these inequalities across all the rooms we get 
	\begin{align}
	\label{equation:sum-payoffs}
	\sum_{r \in [n+1]}Q_r(x^1,x^2,...,x^{n+1}) \leq -\left(\sum_{r \in [n+1]} \sum_{a \in [n+1]} x^a_r \right) = -\left(n+1\right)
	\end{align}
	Since the sum of the expected payoffs of player $0$ over all the $(n+1)$ rooms is bounded from above by $-(n+1)$, we get that the expected payoff $Q_r(x^1,x^2,...,x^{n+1})$, for any room $r$, is upper bounded by $-1$.
\end{proof}

Using this bound, we will now strengthen \cref{Claim1} and show that at any equilibrium, player $(n+1)$ plays the action corresponding to the extra room $(n+1)$ with probability $1$.

\begin{claim}
	\label{claim:claim4}
	Under any Nash equilibrium $(x^0, x^1, \ldots ,x^n)$ of the polymatrix game we have $x^{n+1}_{n+1}=1$.
\end{claim}
\begin{proof}
	
	
	We know (by \cref{Claim1}) that $x^{n+1}_{n+1} > 0$ and, hence, action $(n+1)$ is contained in Supp$(x^{n+1})$. Say, for contradiction, that $x^{n+1}_{n+1} < 1$. This implies that there exists another room $r \in [n] \cap \text{Supp}(x^{n+1})$. Hence, player $(n+1)$ gets the same payoff from the rooms $r$ and $(n+1)$. By construction of the payoff matrices, this is possible only if $x^0_r = 1/3n$, which implies that the associated price $p_r$ of room $r$ is equal to $0$. Now, comparing the expected payoffs for rooms $(n+1)$ and $r$ for every other player $a \in [n]$, we get the following inequality $v_a(r,0) > v_a(n+1,M(1 - 3n \ x^0_{n+1})) = \max_{r \in [n]} v_a(r,M) + \eta$.
	
	Therefore, room $(n+1)$ is not a best response for any player $a \in [n]$, and we must have $x^a_{n+1} =0$ for all $a \in [n]$. Note that in this case the expected payoff of player $0$ for action $(n+1)$ is exactly equal to $-x^{n+1}_{n+1}$ which, by assumption, is strictly greater than $-1$. This contradicts \cref{claim:claim3} and, by way of contradiction, the stated claim follows. 
\end{proof}

Note that the above claim implies that at any equilibrium, player $(n+1)$ plays the action corresponding to any room $r \in [n]$ with zero probability i.e. $x^{n+1}_r = 0$ for all rooms $r \in [n]$. The final claim establishes that the probability mass on the remaining rooms $r \in [n]$ is exactly equal to one. 
\begin{claim}
	\label{claim:column-sum}
	Let $(x^0, x^1, \ldots ,x^n)$ be a Nash equilibrium of the game $\G$. Then, for all rooms $r \in [n]$, we have $ \sum_{a \in [n]} x^a_r  =1$.
\end{claim}

\begin{proof}
	We begin by proving that no player $a \in [n]$ considers action $(n+1)$ as a best response to $x^0$ at a Nash equilibrium. Note that if a player $a \in [n]$ plays action $(n+1)$ with nonzero probability ($x^a_{n+1} >0$), then $Q_{n+1}(x^1,x^2,...,x^{n+1}) < -1$; see \cref{claim:claim4} and \cref{colpayoff}.  
	
	Since player $0$'s expected payoff is the same for every room (\cref{claim:claim3}), it must be the case that $Q_r(x^1,x^2,...,x^{n+1}) < -1 $ for all $r$. 
	Adding up these inequalities we get $\sum_{r \in [n+1]} Q_r(x^1,x^2,...,x^{n+1}) < -(n+1)$.
	
	However, summing up the expected payoff of player $0$ across all actions (and substituting $x^{n+1}_r = 0$ for all $r \in [n]$), gives the following equality
	\begin{align}
	\label{equality:sumofq}
	\sum_{r \in [n+1]} Q_r(x^1,x^2, \ldots, x^{n+1}) = -\left( \sum_{r \in [n]} \sum_{a \in [n+1]} x^a_r \right) - \sum_{a \in [n+1]} x^a_{n+1} = -(n+1).
	\end{align}
	This leads to a contradiction. Therefore, it must be the case that $x^a_{n+1} = 0$ for all $a \in [n]$. These equalities, \cref{claim:claim4}, and  \cref{equality:sumofq} gives us $ \sum_{a \in [n]} \sum_{r\in [n]}  x^a_r  = n$.     
	
	Finally, the equality of payoffs for player $0$ (in particular, $Q_r(x^1,\ldots, x^{n+1}) = Q_{r'}(x^1, \ldots, x^{n+1})$ for all $r, r' \in [n]$) establishes the desired claim $ \sum_{a \in [n]} x^a_r  =1$.
\end{proof}


Using these claims, we will now prove \cref{ppadreduction}. Specifically, we will show that, given a Nash equilibrium $(x^0, x^1, \ldots ,x^n)$ of the game, one can efficiently find an allocation $\pi$ for players $1$ to $n$ (i.e., for the $n$ agents of the underlying rent-division instance) and a price vector $p \in \mathbb{R}_+^n$ (from the mixed strategy $x^0$ of player $0$) such that $(\pi,p)$ is an $\EF$ solution of rent-division instance $\mL$.

For each room $r \in [n]$, set $p_r= M(1-3n \ x^0_r)$. The resulting vector $p \in \mathbb{R}^n$ is the price vector. Let $X$ be the $n \times n$ matrix in which the $n$ rows are equal to the vectors (mixed strategies) $x^1, x^2,\ldots, x^n$. For all $a \in [n]$ we have $\sum_{r \in [n]} x^a_r =1$ as $x^a \in \Delta^n$ and $x^a_{n+1} = 0$. Furthermore, by \cref{claim:column-sum}, we know that $ \sum_{a \in [n]} x^a_r  =1$ for all rooms $r \in [n] $. Therefore, $X$ is a doubly stochastic matrix. Write $\pi$ to denote a permutation in the support of the doubly stochastic matrix $X$.

Note that for each agent $a \in [n]$ we have $x^a_{\pi(a)} > 0$. Hence, action/room $\pi(a)$ is a best response of player $a$ against mixed strategy $x^0$ of player $0$. The payoffs and the best-response property ensures that the following inequalities hold for all rooms $r \in [n]$: $v_a(\pi(a), M(1-3n \ x^0_{\pi(a)})) \geq v_a(r, M(1-3n \ x^0_r))$. The definition of the price vector $p$ gives us envy freeness: $ v_a(\pi(a),p_{\pi(a)}) \geq  v_a(r,p_r) $ for all $a$ and $r$.

Note that in any Nash equilibrium of the constructed game, action $(n+1)$ is a best response of player $(n+1)$ (\cref{claim:claim4}). Therefore, conditions \cref{threshold} ensure that for any room $r\in [n]$ we have $x^0_r \leq 1/3n$, i.e., the price $p_r = M(1-3n \ x^0_r)$ of any room $r \in [n]$ is nonnegative.  

This completes the proof that, given a Nash equilibrium of the polymatrix game $\G$, an $\EF$ solution $(\pi, p)$ of rent-division instance $\mL$ can be computed efficiently.  


\subsubsection{Piecewise Linear to Linear}

The following lemma provides a reduction from $\plc$ to $\linear$. Given a rent-division instance $\I$ of $\plc$ we construct an instance $\mathcal{L}$ of $\linear$ by considering the linear part of the utility functions of $\I$ beyond the last breakpoint. We further show that an envy-free solution of $\mathcal{L}$ can be efficiently modified to produce an envy-free solution for $\I$.

\begin{lemma} 
	\label{plc-to-linear}
	Any instance $\I$ of $\plc$ can be efficiently reduced to an instance $\mL$ of $\linear$ such that, given an $\EF$ solution of $\mL$, we can find an $\EF$ solution of $\I$ in polynomial time. 
\end{lemma}

\begin{proof}
	Let $ \I = \langle \A, \R, \{v_a(r, \cdot) \}_{a, r}  \rangle$ be an instance of $\plc$. Write $K \in \mathbb{R}_+$ denote the last breakpoint among all the utility functions $\{v_a(r,.)\}_{a, r}$. That is, each utility function $v_a(r, z)$ is linear for prices $z \geq K$ with a fixed slope, say, $\lambda^a_r$.
	
	We construct an instance $ \mathcal{L} =  \langle \A, \R, \{v'_a(r, \cdot) \}_{a, r}  \rangle$ of $\linear$. Here, each utility function $v'_a(r,z)$ is of the  form $v'_a(r,z) := H^a_r - \lambda^a_r z$ with base value $H^a_r = v_a(r,K)$ (i.e., the base value is equal to the utility at price $K$).
	
	
	Let $(\pi,p')$ be an envy-free solution of $\mathcal{L}$ with $p'_r \ge 0$ for all rooms $r \in \R$. Note that finding such a solution is a $\PPAD$ problem. Construct a price vector $p \in \mathbb{R}^n$ by setting  $p_r \coloneqq  p'_r+K$ for all rooms $r \in \R$. We claim that $(\pi,p)$ is an envy-free solution of the given instance $\I$. Note that the utility functions $\{v_a(r, \cdot)\}_{a, r}$ of instance $\I$ satisfy 
	\begin{align*}
	v_a(r, p_r) & = v_a(r, K+p'_r) \\ 
	& =  v_a(r,K) - \lambda^a_r \ p'_r \\
	& = H^a_r - \lambda^a_r \ p'_r \\
	& = v'_a(r,p'_r)
	\end{align*}
	Therefore, envy-freeness of $(\pi,p')$ for the constructed instance $\mathcal{L}$ directly implies envy-freeness of $(\pi,p)$ for the given instance $\I$.
	
\end{proof}

\cref{ppadreduction} and \cref{plc-to-linear} establish the main result of this section:

\begin{lemma}
	$\plc$ lies in the complexity class $\PPAD$.
\end{lemma}

\subsection{Proof of PLS Containment}
\label{subsection:pls}
The class \PLS (Polynomial Local Search) was defined by Johnson et al.~\cite{johnson1988easy} to capture the complexity of finding local optima of optimization problems. Here, a generic instance $\I$ of an optimization problem has a corresponding finite set of solutions $S(\I)$ and a potential $c(s)$ associated with each solution $s \in S(\I)$. The objective is to find a solution that maximizes (or minimizes) this potential. In the local search version of the problem, each solution $s \in S(\I)$ additionally has a well-defined neighborhood $N(s) \in 2^{S(\I)}$ and the objective is to find a local optimum, i.e., a solution $s\in S(\I)$ such that no solution in its neighborhood $N(s)$ has a higher potential.

\begin{definition}[$\PLS$]
	\label{definition:pls}
	Consider an optimization problem $\mathcal{X}$, and for all input instances $\I$ of $\mathcal{X}$ let $S(\mathcal{I})$ denote the finite set of feasible solutions for this instance, $N(s)$ be the neighborhood of a solution $s \in S(\I)$,  and $c(s)$ be the potential of solution $s$. The desired output is a local optimum with respect to the potential function. 
	
	Specifically, $\mathcal{X}$ is a polynomial local search problem (i.e., $\mathcal{X} \in \PLS$) if all solutions are bounded in the size of the input $\mathcal{I}$ and there exists polynomial-time algorithms $\mathcal{A}_1$, $\mathcal{A}_2$, and $\mathcal{A}_3$ such that:
	\begin{enumerate}[(a)]
		\item $\mathcal{A}_1$ tests whether the input $\mathcal{I}$ is a legitimate instance of $\mathcal{X}$ and if yes, outputs a solution $s_{\text{initial}} \in S(\mathcal{I})$.
		\item $\mathcal{A}_2$ takes as input instance $\mathcal{I}$ and candidate solution $s$, tests if $s \in S(\mathcal{I})$ and if yes, computes $c(s)$.
		\item $\mathcal{A}_3$ takes as input instance $\mathcal{I}$ and candidate solution $s$, tests if $s$ is a local optimum and if not, outputs $s' \in N(s)$ such that $c(s') > c(s)$ (the inequality is reversed for the minimization version).
	\end{enumerate}
\end{definition}

Each $\PLS$ problem comes with an associated local search algorithm that is implicitly described by the three algorithms mentioned above. The first algorithm is used to find an initial solution to the problem and the third algorithm is iteratively used to find a potential-improving neighbor until a local optimum is reached. 


Fair rent division does not directly qualify as a local search problem, since there is no apparent potential and \emph{any} $\EF$ allocation with an appropriate price vector is an acceptable solution. However, considering the execution of \cref{algorithm:envy-free} on an instance $\I = \langle \mathcal{A}, \mathcal{R}, \{ v_a(r, \cdot) \}_{a, r} \  \rangle$ of fair rent division we will identify a local search method. 

\cref{algorithm:envy-free} constructs an instance $\widehat{\mathcal{I}} = \left\langle \mathcal{A}, \mathcal{R}, \{ \widehat{v}_a(r, \cdot) \}_{a,r} \right\rangle$ 
from the given instance $\I$. Specifically, using \cref{equation:bar-to-hat} and a sufficiently large $M \in \mathbb{R}_+$, in the algorithm the utility functions of instance $\widehat{\mathcal{I}}$ are set to be quasilinear beyond $M$. Recall that an $\EF$ solution of $\widehat{\mathcal{I}}$ in which all the prices are less than $M$ is an $\EF$ solution of $\I$, hence we will focus on the constructed instance $\widehat{\mathcal{I}}$ for the rest of this subsection and show that finding such a solution is a $\PLS$ problem. Specifically, we will develop a potential to interpret the execution of \cref{algorithm:envy-free} on $\widehat{\mathcal{I}}$ as a local search algorithm. Here, algorithm $\mathcal{A}_1$ (see \cref{definition:pls}) can be directly implemented, since we can compute an initial $\EF$ solution of instance $\widehat{\mathcal{I}}$ in polynomial time for the quasilinear utilities beyond $M$.



The solution space of any $\PLS$ problem has to be finite. This rules out specifying $\EF$ solutions in the standard form $(\pi,p)$, since the space of such solutions is continuous. However, an alternate way to specify envy-free solutions is through a tuple $\tup$ where $\pi$ is an allocation of the rooms to the agents and $(L,U)$ is a linear domain of $\widehat{\I}$. For any linear domain $(L,U)$, the components of the vectors $L$ and $U$ are the breakpoints of the modified utility functions, hence the total number of linear domains of instance $\widehat{\I}$ is always finite--specifically, it is $\mathcal{O}((\ell+n)^{2n})$, where $n$ is the number of agents and $\ell$ is the total number of breakpoints of the given instance $\I$. Let $\E\tup$ be a set of vectors $x \in \mathbbm{R}^n$ that satisfy the following linear constraints: (i) $L \leq x \leq U $, (ii) $\widehat{v}_a( \pi(a) , x_{\pi(a)}) \geq \widehat{v}_a(r, x_r) $ for all $a$ and $r$, (iii) $x \geq \mathbb{0}$.


Since $(L,U)$ is a linear domain, for all price vectors $x$ in $(L,U)$ the utility functions $\{ \widehat{v}_a(r, x) \}_{a,r}$ are linear in $x$. Hence, $\E\tup$ is a polytope. The tuple $\tup$ is accepted as a valid solution of the local search problem if the polytope $\E\tup$ is nonempty, a property that is verifiable in polynomial time. Note that every price vector $p \in \E\tup$ can be paired with the allocation $\pi$ to generate an $\EF$ solution of $\widehat{\mathcal{I}}$. In particular, we are interested in the price vector $p\tup$ that is defined to be the solution to the following linear program:
\begin{align*}
\min_{x \in \mathbbm{R}^n} \ \  \sum_r x_r &  \ \  \ \ \  \text{ subject to} \ \  x \in \E(\pi,L,U)
\end{align*}	

Note that it is not necessary that $(L,U)$ is the linear domain of price vector $p\tup$. Recall that $(L,U)$ is the linear domain of price vector $p$ if and only if for all rooms $r$ we have $ L_r < p_r \le U_r$. Therefore, if there exists a room $r$ such that $p_r\tup = L_r$ (which is not ruled out by the LP constraints), then $(L,U)$ is not the linear domain of price vector $p\tup$. 


We define two potential functions, one for the allocations within a linear domain and the other for the linear domains themselves, and then combine these functions in an appropriate manner to generate the potential for the fair rent division problem. \\

\noindent 
{\bf Potential Function for Linear Domains:} The proof of \cref{theorem:ef-rounded} employs a counting argument over the linear pieces of the utility functions $\{ \widehat{v}_a(r, \cdot) \}_{a,r}$ to track the progress of \cref{algorithm:envy-free}. The same idea is used to define a potential function over the linear domains.

To begin with, we rewrite each linear domain $(L,U)$ as an $n$-tuple of integers. For each $a$ and $r$, let $q^a_r \in \mathbb{Z}$ denote the number of linear pieces (i.e., the number of breakpoints minus one) of the function $\widehat{v}_a(r,\cdot)$  between $0$ and $L_r$. For each room $r$, write $q_r$ to denote the quantity $\sum_a q^a_r$. Note that each component $L_r$ of the lower bound is a breakpoint of one of the functions $\{v_a(r, \cdot)\}_{a \in \mathcal{A}}$. Hence, considering all such breakpoints in a sorted  order, we can use the integer $q_r$ to uniquely identify $L_r$ from the sorted list of breakpoints. In addition, observe that a linear domain can be uniquely identified using only its lower bound $L$. This follows from the the maximal property of linear domains; each upper bound $U_r$ is the first breakpoint that appears after $L_r$ among all the functions $\{v_a(r, \cdot)\}_{a \in \mathcal{A}}$. Thus, each linear domain $(L,U)$ can be equivalently represented by an $n$-tuple of integers $(q_1,q_2, \ldots, q_n)$. Let $\ell^a_r$ represent the total number of linear pieces of the piecewise-linear function $\widehat{v}_a(r,\cdot)$ and write $\ell_r \coloneqq \sum_a \ell^a_r$. Using these constructs we define the potential function for linear domains as $\phi_1(L,U) \coloneqq \sum_r (\ell_r - q_r)$.

This potential increases as the lower bounds of the linear domains move closer to $0$. Another relevant property of $\phi_1$  is that its value is always an integer and any increase in $\phi_1$ is at least $1$. \\

\noindent
{\bf Potential Function for Allocations:} Consider a tuple $\tup$ and let $w(\pi)$ denote the weight of the matching $\pi$ in the complete bipartite graph between agents and rooms where the weight of edge $(a,r)$ is the logarithm of the slope of the function $\widehat{v}_a(r,\cdot)$ in the linear domain $(L,U)$. Let $\Lambda^+$ and $\Lambda^-$ be the smallest and largest slopes across all the piecewise linear utility functions $\{ \widehat{v}_a(r, \cdot) \}_{a,r}$. This leads to the following inequalities: $n \log \left(\Lambda^-\right) \leq w(\pi) \leq n \log\left( \Lambda^+ \right)$.

Now, we define the potential functions for the allocations:
\begin{align*}
\phi_2\tup := \frac{w(\pi) - n\log(\Lambda^-)}{n\log(\Lambda^+) - n\log(\Lambda^-) +1}
\end{align*}
The normalization ensures that for all inputs $\tup$, the potential satisfies $\phi_2\tup \in [0,1)$. This property ensures that any improvement in the first potential strictly makes up for any loss in the second potential. Also, note that this potential is strictly increasing as a function of $w(\pi)$. Hence,  \cref{lemma:wts-drop} implies that potential $\phi_2$ can be used to track the progress of \cref{algorithm:envy-free} over consecutive allocations within a linear domain. 

These two potentials are combined to obtain the potential for the local search version of fair rent division.
\begin{align*}
\phi\tup &:= \begin{cases}
\sum_r \ell_r+1  & \text{ if } p_r\tup  = 0 \text{ for some } r\\
\phi_1(L,U)+ \phi_2\tup & \text{otherwise}
\end{cases}
\end{align*}

The solution $\tup$ whose associate envy-free price vector $p\tup$ has a zero component is accorded the maximum possible value of potential. This choice reflects the fact that a solution $(\pi,p\tup)$ with a zero price component is not just only envy free for $\widehat{\I}$, but it is also envy free for $\I$ (a detailed explanation of this fact appears in the proof of \cref{theorem:ef-rounded}). Therefore, every optimal solution of the potential $\phi$ encodes an $\EF$ solution of the rent-division instance $\I$. We now prove that the problem of finding a fair solution lies is $\PLS$ by providing for it a local search algorithm that operates on the potential $\phi$.

\begin{lemma}
	$\plc$ is in $\PLS$.
\end{lemma}
\begin{proof}
	Given instance $\I = \langle \mathcal{A}, \mathcal{R}, \{ v_a(r, \cdot) \}_{a, r} \  \rangle$ of $\plc$, let  $\widehat{\mathcal{I}} = \left\langle \mathcal{A}, \mathcal{R}, \{ \widehat{v}_a(r, \cdot) \}_{a,r} \right\rangle$ be the instance constructed using \cref{equation:bar-to-hat} with a sufficiently large threshold 
	$M$. 
	
	The set of candidate solutions $S(\widehat{\I})$ consists of all the tuples $\tup$ such that $\E\tup$ is non-empty. In this setup, the three algorithms detailed in \cref{definition:pls} can be implemented efficiently. For algorithm $\mathcal{A}_1$, using the fact that the utility functions are quasilinear above $M$, an initial solution can be efficiently computed by using the algorithm of Aragones~\cite{aragones1995derivation}. We can implement algorithm $\mathcal{A}_2$ efficiently as well, since one can verify whether a proposed solution $\tup$ is contained in $S(\widehat{\I})$ in polynomial time; this step entails solving a linear program. In addition, the potential function $\phi$ described above can be efficiently computed for every input $\tup$. 
	
	There are two cases to consider for the neighborhood algorithm $\mathcal{A}_3$. In the first case, there exists room $r$ such that $p_r\tup = 0$. Since the potential of such an input is the maximum possible value, this input is declared as a local optimum. The second complementary case splits into two subcases. \\
	
	\noindent 
	\emph{Subcase 1:} $(L,U)$ is the linear domain of price vector $p\tup$. In this case the algorithm $\mathcal{A}_3$ outputs a solution/neighbor $(\sigma, L,U)$, where $\sigma$ is a maximum weight perfect matching in the graph $\mathcal{F}^w(p\tup)$. Note that $\sigma$ can be computed deterministically in polynomial time. The tuple $(\sigma, L,U)$ is a legitimate solution (i.e., $(\sigma, L,U) \in S(\widehat{\mathcal{I}})$) since the polytope $\E(\sigma, L,U)$ is nonempty--the price vector $p\tup$, in particular, belongs to it. Next will show that $(\sigma, L,U)$ is a solution with larger potential: since $(L,U)$ is the linear domain of $p\tup$, one can invoke \cref{lemma:wts-drop} to prove that $w(\sigma) > w(\pi)$. This implies that $ \phi_2(\sigma, L,U) > \phi_2\tup$. Therefore, $\phi(\sigma, L,U) > \phi\tup$. Overall, for solutions $p\tup$ satisfying this case algorithm $\mathcal{A}_3$ finds a neighbor with a higher potential. \\
	
	\noindent 
	\emph{Subcase 2:} $(L,U)$ is not the linear domain of price vector $p\tup$. Write $(L', U')$ to denote the linear domain $p\tup$. Here, the neighbor computed by algorithm $\mathcal{A}_3$ is $(\pi, L',U')$. This is a valid solution (i.e., $(\pi, L', U') \in S(\widehat{\mathcal{I}})$), since price vector $p\tup$ belongs to the polytope $\E(\pi, L', U')$, making the polytope nonempty.  
	
	The potential $\phi$ increases at this neighboring solution. Specifically, let the $n$-tuples $(q_1,q_2,...,q_n)$ and $(q_1',q_2',...,q_n')$ represent the linear domains $(L,U)$ and $(L',U')$, respectively. Since  $L \leq p\tup \leq U$, we have $L'_i \leq L_i$ and, hence, $q_i' \leq q_i$ for all $i \in [n]$. Furthermore, for $p\tup$ to not be contained in the linear domain $(L,U)$, there exists at least one room $r$ such that $p_r\tup = L_r$ implying that $L_r' < L_r$, i.e., $q_r' =q_r -1$. Therefore, comparing the tuples representing the linear domains $(L,U)$ and $(L',U')$, we get $\sum_r q'_r +1 \leq  \sum_r q_r$. This gives us the following inequality between the potential values of the linear domains $\phi_1(L',U') \geq \phi_1(L,U) +1$.
	
	With a change in linear domains, the associated slopes (and the weight $w$ of the allocation) of the utilities change as well. However, we know that the value of the potential $\phi_2$ (which depends on $w$) is always within $[0,1)$. Hence, we have $\phi_2(\pi, L',U') > \phi_2\tup -1$. The above mentioned inequalities give us the desired increase in potential $\phi(\pi, L',U') > \phi\tup$.
\end{proof}

\bibliographystyle{alpha}
\bibliography{references}
\appendix
\section{Rent Division with Quasilinear Utilities} 
\label{section:quasi}
The work of Aragones~\cite{aragones1995derivation} provides a characterization of the allocations that constitute envy-free solutions under quasilinear utilities and, using this characterization,  develops an efficient algorithm for finding fair solutions for the quasilinear setting. For completeness, this section details this characterization and presents a polynomial-time algorithm for finding envy-free solutions in this setting. 

Consider an instance $\mathcal{Q} =  \langle \mathcal{A}, \mathcal{R}, \{ v_a(r, \cdot) \}_{a\in \mathcal{A}, r \in \mathcal{R}} \  \rangle$ wherein $|\mathcal{A}| = |\mathcal{R}| =n$ and each utility function $v_a(r, \cdot)$ is quasilinear, i.e., is of the form $ v_a(r,x) = B^a_r - x$. Here, $B^a_r$ is agent $a$'s base value for the room $r$. In other words, agent $a$'s utility for room $r$ is the base value $B^a_r$ minus the price/rent of room $r$. In addition, write $\mathcal{G}(\mathcal{Q}) = (\mathcal{A} \cup \mathcal{R}, \mathcal{A} \times \mathcal{R})$ to denote the complete bipartite graph (between the agents and the rooms) in which the weight of each edge $(a,r)$ is equal to $B^a_r$. 

The following lemma asserts that allocations that constitute envy-free solutions of $\mathcal{Q}$  correspond to maximum weight perfect matchings in $\mathcal{G}(\mathcal{Q})$.


\begin{lemma}
	\label{lemma:ef-quasi}
	For any rent-division instance $\mathcal{Q} =  \langle \mathcal{A}, \mathcal{R}, \{ v_a(r, \cdot) \}_{a , r }  \rangle$ with quasilinear utilities (i.e., $ v_a(r,x) = B^a_r - x$ for all $a$ and $r$), an allocation $\pi$ can be coupled with a price vector to realize an envy-free solution of $\mathcal{Q}$ if and only if  $\pi$ is a maximum weight perfect matching in the weighted graph $\mathcal{G}(\mathcal{Q})$. 
\end{lemma}
\begin{proof}
	
	Since quasilinear utilities are a special case of piecewise-linear utilities, the existence results mentioned previously ensure that $\mathcal{Q}$ admits an $\EF$ solution. In fact, linear-programming duality can be used to provide a stand-alone proof of existence for the quasilinear setting.  We will assume the existence of fair solutions in this setting and first show that if $(\pi, p)$ is an envy-free solution, then for any maximum weight perfect matching, $\sigma$, the solution $(\sigma, p)$ is also envy free. \\
	
	\noindent 
	{\bf If Part:} Since $(\pi,p)$ is an $\EF$ solution, $v_a(\pi(a),p_{\pi(a)}) \geq v_a(\sigma(a),p_{\sigma(a)})$ for all agents $a$. Given that $\sigma$ is  maximum weight perfect matching in $\mathcal{G}(\mathcal{Q})$, we have $ \sum_a B^a_{\pi(a)}  \leq \sum_a B^a_{\sigma(a)}$. Subtracting $\sum_r p_r$ from both sides gives us $\sum_a B^a_{\pi(a)} - \sum_r p_r  \leq \sum_a B^a_{\sigma(a)} - \sum_r p_r$. Regrouping the terms, we get $\sum_a (B^a_{\pi(a)} -p_{\pi(a)}) \leq \sum_a (B^a_{\sigma(a)} -p_{\sigma(a)})$. That is, $\sum_a v_a(\pi(a),p_{\pi(a)}) \leq \sum_a v_a(\sigma(a),p_{\sigma(a)})$.
	
	
	Hence, the termwise inequality (which follows from the envy freeness of $(\pi, p)$) mentioned above implies that $v_a(\pi(a),p_{\pi(a)}) = v_a(\sigma(a),p_{\sigma(a)})$ for all agents $a$. Therefore, $(\sigma, p)$ is also an envy-free solution of the given instance $\mathcal{Q}$. \\
	
	\noindent 
	{\bf Only if Part:} Assume, for contradiction, that $\pi$ is not a maximum weight perfect matching, but $(\pi, p)$ is an $\EF$ solution. Write $\pi^*$ to denote a maximum weight perfect matching in $\mathcal{G}(\mathcal{Q})$. Therefore,
	\begin{align*}
	\sum_a B^a_{\pi(a)} &< \sum_a B^a_{\pi^*(a)} \\
	\sum_a B^a_{\pi(a)} - \sum_r p_r &< \sum_a B^a_{\pi^*(a)} - \sum_r p_r  &\text{(subtracting $\sum_r p_r$)} \\
	\sum_a (B^a_{\pi(a)} -p_{\pi(a)}) &< \sum_a (B^a_{\pi^*(a)} -p_{\pi^*(a)})   \\
	\sum_a v_a(\pi(a),p_{\pi(a)}) &< \sum_a v_a(\pi^*(a),p_{\pi^*(a)})
	\end{align*} 
	Hence, there exists an agent $a$ such that $v_a(\pi(a),p_{\pi(a)}) < v_a(\pi^*(a),p_{\pi^*(a)})$. This inequality contradicts the envy freeness of $(\pi, p)$.

\end{proof}

A direct implication of \cref{lemma:ef-quasi} is that finding \emph{any} maximum weight perfect matching of $\mathcal{G}(\mathcal{Q})$ (which is a polynomial-time computation) gives us an allocation of an envy-free solution. Once the allocation is fixed, for quasilinear utilities the envy-free requirements correspond to  linear constraints. Therefore, we can write a linear program to find envy-free prices. 

{
	\begin{algorithm}[H]
		{
			\caption{Algorithm to compute an envy-free solution under quasilinear utilities}
			{\bf Input:} A rent division instance with quasilinear utilities $\mathcal{Q} =  \langle \mathcal{A}, \mathcal{R}, \{ v_a(r, \cdot) \}_{a , r } \  \rangle$ \\
			{\bf Output:} An envy-free solution $(\pi,p)$ for $\mathcal{Q}$
			\label{algorithm:quasi-linear}
			\begin{algorithmic}[1]
				\STATE Compute $\pi$, a maximum weight perfect matching in $\mathcal{G}(\mathcal{Q})$ 
				\STATE Set $p$ to be a solution of the following linear program
				\begin{align*}
				\max_{x \in \mathbb{R}^n} & \ \  \mathbb{0}^T x \\  
				\text{subject to} & \ \ \qquad\qquad x  \geq \mathbb{0}\\ 
				& \ \ \ v_a( \pi(a) , x_{\pi(a)}) \geq v_a(r, x_r) \quad \text{ for all } a \in \mathcal{A}  \text{ and } r \in \mathcal{R} 
				\end{align*}
				\STATE Return $(\pi, p)$
			\end{algorithmic}
		}
	\end{algorithm}
}

\begin{lemma}
	\label{theorem:quasi-linear}
	Given any rent division instance $\mathcal{Q} =  \langle \mathcal{A}, \mathcal{R}, \{ v_a(r, \cdot) \}_{a, r} \  \rangle$ with quasilinear utilities, \cref{algorithm:quasi-linear} computes an envy-free solution of $\mathcal{Q}$ in polynomial time.
\end{lemma}
\begin{proof}
	The runtime analysis is direct. A maximum weight perfect matching can be computed in polynomial time. Furthermore, the linear program can be solved in time polynomial in $n$, since it contains $n$ variables and $\mathcal{O}(n^2)$ constraints. 
	
	The correctness of the algorithm follows from \cref{lemma:ef-quasi}.
\end{proof}

\section{Perturbation Lemma}
\label{appendix:perturbation}

As mentioned previously, this works uses a variant of the Perturbation Lemma of Alkan et al.~\cite{alkan1991fair} (\cref{lemma:perturb}).  This section provides a proof of this variant for completeness. 

We first consider rent-division instances $\mathcal{Q} = \langle \mathcal{A}, \mathcal{R}, \{ v_a(r, \cdot) \}_{a, r} \  \rangle$ with quasilinear utilities, $v_a(r, x) := B^a_r - x$. Write $\mathcal{G}(\mathcal{Q}) = (\mathcal{A} \cup \mathcal{R}, \mathcal{A} \times \mathcal{R})$ to denote the complete bipartite graph (between the agents and the rooms) in which the weight of each edge $(a,r) \in \mathcal{A} \times \mathcal{R}$ is equal to $B^a_r$. 

Recall that \cref{lemma:ef-quasi} shows that if $\pi$ is a maximum weight perfect matching in $\mathcal{G}(\mathcal{Q})$, then there exists a price vector $p \in \mathbb{R}^n$ such that  $B^a_{\pi(a)} - p_{\pi(a)} \geq B^a_r - p_r$ for all $a$ and $r$ (i.e., $(\pi, p)$ is an $\EF$ solution). Next, we state and prove a multiplicative version of this lemma.  




\begin{lemma} \label{multiplicative quasi}
	Let $\mathcal{H} =([n] \cup [n], [n] \times [n], W) $ be a complete bipartite graph with nonnegative edge weights, $W_{(a, r)} \geq 0$ for all $a, r \in [n]$. If $\mathcal{H}$ admits a positive matching (i.e., there exists a matching $\pi$ such that $ W_{(a, \pi(a))} > 0$ for all $ a \in [n]$), then there exists a matching $\sigma$ of $\mathcal{H}$ and a positive vector $d \in \mathbb{R}^n_+$ such that 
	\begin{align*}
	d_{\sigma(a)} W_{(a, \sigma(a))} & \geq d_r W_{(a, r)} \qquad \text{for all $a,r \in [n]$}
	\end{align*}
\end{lemma}

\begin{proof}
	Construct a bipartite graph $\mathcal{G}=(\mathcal{A} \cup \mathcal{R}, \mathcal{A} \times \mathcal{R})$ with $\mathcal{A}=\mathcal{R}=[n]$ and edge weights $\{ B^a_r \}_{a,r}$ as follows
	\begin{align*}
	B^a_r  & :=
	\begin{cases}
	\log W_{(a, r)}  & \quad \text{if } W_{(a, r)} > 0 \ \text{ for } (a,r) \in \mathcal{A} \times \mathcal{R}  \\
	-\infty & \quad \text{else if} \ \ \ W_{(a,r)} = 0 
	\end{cases} 
	\end{align*}
	
	Write $\sigma$ to denote a maximum weight perfect matching in $\mathcal{G}$. Since $\mathcal{H}$ contains a positive perfect matching,  the weight of $\sigma$ is finite. In addition, \cref{lemma:ef-quasi} ensures that there exists $p \in \mathbb{R}^n$ such that 
	\begin{align*}
	\log W_{(a, \sigma(a))} - p_{\sigma(a)} & \geq \log W_{(a, r)} - p_r \qquad \text{ for all } a, r \in [n].
	\end{align*}
	
	Setting $d_r := \exp(-p_r) > 0$ and taking exponentials, we get the stated result. 
\end{proof}

We are now ready to prove the variant of Perturbation Lemma which identifies the allocations (bijections) which are realizable in the perturbed solutions. 

Recall that $\mathcal{F}(p)$ denotes the first-choice graph under a given price vector $p \in \mathbb{R}^n$, and $\mathcal{F}^w(p)$ denotes the weighted version on this bipartite graph. Also, note that in the following lemma the utilities are not confined to be quasilinear. 

\LemmaStatementPerturb*


\begin{proof}
	Let $(L,U)$ denote the linear domain of the given price vector $p$ and write $F$ to denote the edge set of the graph $\mathcal{F}(p)$, i.e.,  $F := \{ (a,r) \in \A \times \R \mid  v_a(r, p_r) = v_a(\pi(a), p_{\pi(a)}) \}$. Note that $(a, \pi(a))$ belongs to this set of first-choice edges ${F}$ for all agents $a \in \mathcal{A}$. 
	
	As before, we use $w_{(a,r)}$ to denote the edge weights of the graph $\mathcal{F}^w(p)$, i.e., $w_{(a,r)} := \log \lambda^a_r$,  for all $(a,r) \in F$. Here, $\lambda^a_r >0 $ is the (fixed) magnitude of the slope of the utility function $v_a(r, \cdot)$ in the linear domain $(L, U)$.
	
	Since the utility functions $v_a(r, \cdot)$ are continuous, monotone decreasing, and piecewise linear, the following equality holds for all $a \in \A$, $r \in \R$, and any vector $d \in \mathbb{R}^n_+$ with sufficiently small, positive components 
	\begin{align*}
	v_a(r, p_r-d_r) & = v_a(r, p_r) + \lambda^a_r \ d_r 
	\end{align*}
	
	Define a complete bipartite graph $\mathcal{H} = (\mathcal{A} \cup \mathcal{R}, \mathcal{A} \times \mathcal{R}, W)$ wherein the weight of each edge $(a,r) \in F$ is set to be $W_{(a,r)} = \exp( w_{(a,r)} ) = \lambda^a_r >0$, and $W_{(a,r)}= 0$ for all $(a,r) \in (\mathcal{A} \times \mathcal{R}) \setminus F$. The edge weights in $\mathcal{H}$ are nonnegative and it contains a positive matching, in particular $\pi$; specifically, only the first-choice edges (i.e., the edges in $F$) have positive weight in $\mathcal{H}$. 
	
	Lemma \ref{multiplicative quasi} implies that there exists a matching $\sigma$ and a positive vector $d \in \mathbb{R}_+^n$ such that 
	
	\begin{align}
	\label{ineq:mult-factor}
	d_{\sigma(a)}  \ W_{(a, \sigma(a))} & \geq d_r \  W_{(a,r)} \quad \text{ for all $a$ and $r$}.
	\end{align}
	
	It is relevant to note that the graph $\mathcal{G}$ considered in the proof of \cref{multiplicative quasi} corresponds to $\mathcal{F}^w(p)$. Hence, via the proof of \cref{multiplicative quasi}, we get that $\sigma$ is a maximum weight perfect matching in $\mathcal{F}^w(p)$. 
	
	Furthermore, for each agent  $a \in \A$, the edge $(a, \sigma(a))$ is a first-choice edges, i.e., $(a, \sigma(a)) \in F$: assume, for contradiction, that $(a, \sigma(a)) \notin F$, then the weight of this edge in $\mathcal{H}$ is zero, $W_{(a, \sigma(a))} =0$. Setting $r=\pi(a)$ and using (\ref{ineq:mult-factor}), we get $W_{(a, \pi(a))} = 0$. This contradicts that fact that $(a, \pi(a)) \in F$, since the weight (under $W$) of every edge in $F$ is positive.  
	
	Since $(a, \sigma(a)) \in F$ for all $a$, we have 
	\begin{align} 
	\label{multieq}
	d_{\sigma(a)}  \ \lambda^a_{\sigma(a)} &  \geq d_r \  \lambda^a_r  \quad \text{for all } (a, r) \in F.
	\end{align}
	
	Note that the components of the positive vector $d$ can be multiplicatively scaled to obtain $0 < d_r \leq \delta$, for any positive parameter $\delta >0$. A uniform scaling also ensures that inequality (\ref{multieq}) continues to hold. 
	
	Setting parameter $\delta >0$ to be small enough and price vector $q := p - d$, next we will complete the proof by establishing that $(\sigma, q)$ is an $\EF$ solution of the given instance. 
	
	For each agent $a$, we will consider two exhaustive cases (i) $(a,r) \in F$ and (ii) $(a,r) \notin F$. 
	If $(a,r) \in F$ (i.e., $(a, r)$ is a first choice edge), then 
	\begin{align*}
	v_a(r, q_r) & = v_a(r, p_r - d_r) & \text{(by definition of $q$)} \\
	& = v_a(r, p_r) + d_r \lambda^a_r & \text{(by peicewise linearity of $v_a(r, \cdot)$ and $d_r \leq \delta$)}\\
	& \leq v_a(r, p_r) + d_{\sigma(a)}  \ \lambda^a_{\sigma(a)} & \text{(by inequality (\ref{multieq}))} \\
	& = v_a(\sigma(a), p_{\sigma(a)}) + d_{\sigma(a)}  \ \lambda^a_{\sigma(a)} & \text{(since $(a, \sigma(a)), (a, r) \in F$)} \\
	& =  v_a(\sigma(a), p_{\sigma(a)} - d_{\sigma(a)}) & \text{(by peicewise linearity of $v_a(r, \cdot)$ and $d_{\sigma(a)} \leq \delta$)}\\
	& =  v_a(\sigma(a), q_{\sigma(a)})& \text{(by definition of $q$)} 
	\end{align*} 
	
	On the other hand, if $(a,r) \notin F$, then we have $v_a(r, p_r)  <  v_a (\pi(a), p_{\pi(a)}) = v_a(\sigma(a), p_{\sigma(a)})$ (recall that $(a, \sigma(a)) \in F$). Therefore, for a small enough $\delta >0$ the following inequality holds $v_a(r, p_r) + d_r \lambda^a_r \leq v_a(\sigma(a), p_{\sigma(a)}) + d_{\sigma(a)} \lambda^a_{\sigma(a)}$.\footnote{Here, $d_r, d_{\sigma(a)} \leq \delta$. Also, note that even if $(a,r) \notin F$, the magnitude of the slope $\lambda^a_r$ (of the utility function $v_a(r, \cdot)$) is well-defined and fixed though the linear domain $(L,U)$.} Since the utility functions are piecewise linear, we get $v_a(r, p_r - d_r) \leq  v_a( \sigma(a), p_{\sigma(a)} - d_{\sigma(a)})$. The definition of $q$ gives us the desired inequality
	
	\begin{align*}
	v_a(r, q_r) \leq v_a(\sigma(a), q_{\sigma(a)})
	\end{align*}

	Overall, we get that $(\sigma, q)$ is an $\EF$ solution wherein $\sigma$ is a maximum weight perfect matching in $\mathcal{F}^w(p)$ and $q$ is componentwise less than the given price vector $p$. This completes the proof. 
\end{proof}

  \section{Counterexamples}
  \label{section:counter}
  This section provides examples to show that  both continuity and boundedness are necessary assumptions on the utility functions to guarantee the existence of an envy-free solution. \\


  \noindent
  {\bf Unbounded utility functions:} Consider a rent-division instance $\I$ with two agents that have identical utilities for each of the two rooms $\{r_1, r_2\}$. In particular, for $a \in \{1, 2\}$, let $v_a(r_1, x) := e^{-x} +2$ and $v_a(r_2, x) :=   -e^x +2$. Therefore, for all $a$ and $r$, $v_a(r, \cdot) $ is continuous and monotone decreasing, but also \textit{unbounded}. Specifically, for any price vector $p \in \mathbb{R}^2$, we have $v_a(r_1, p_1) \neq v_a(r_2, p_2)$. 
  
  Note that in a rent-division instance wherein the utilities of agents are identical, a price vector $p$ induces an envy-free solution if and only if under $p$ the utilities of the agents are equal for all the rooms. Therefore, this instance does not admit an envy-free solution. \\
  
  \noindent 
  {\bf Discontinuous utility functions:}  Consider a rent-division instance $\I$ with two agents that have identical utilities for each of the two rooms $\{r_1, r_2\}$. Specifically, the utility function $v_a(r_1, \cdot)$ is a union of monotone decreasing, linear functions in the $1/10$-neighborhood of even integers and the utility function $v_a(r_2, \cdot)$ is a union of monotone decreasing, linear functions in the $1/10$-neighborhood of odd integers; see Figure \ref{fig:disc}. Therefore, for all $a$ and $r$, $v_a(r, \cdot) $ is monotone decreasing and bounded, but also \textit{discontinuous}. By construction, we have  $v_a(r_1, p_1) \neq v_a(r_2, p_2) $ at any price vector $p \in \mathbb{R}^2$. 
  
  Note that in a rent-division instance in which the utilities of agents are identical, a price vector $p$ induces an envy-free solution if and only if under $p$ the utilities of the agents are same for all the rooms. Therefore, this instance does not admit an envy-free solution.
  \begin{figure}[h]
  	\begin{center}
  		\includegraphics[scale=0.6]{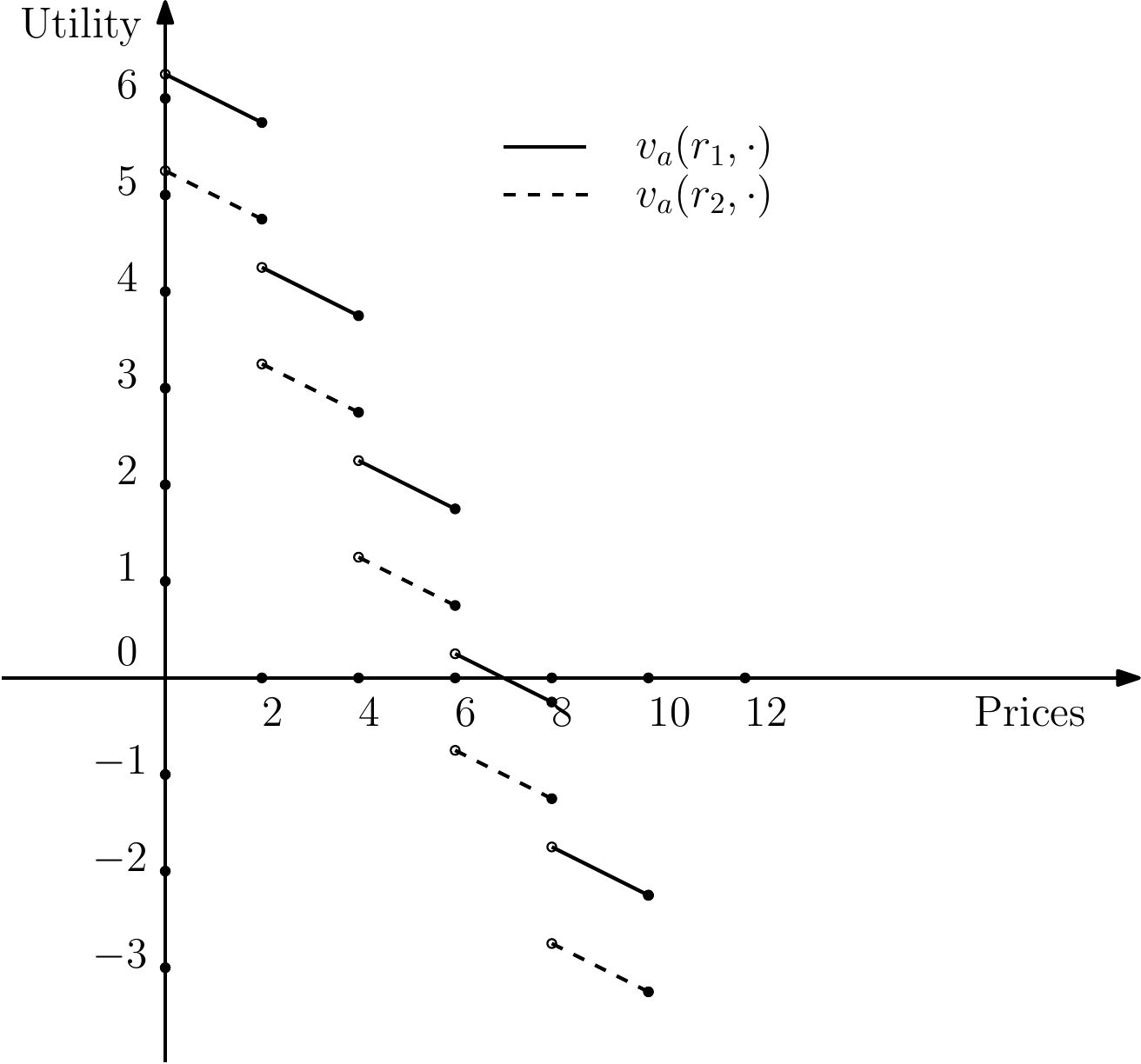}
  	\end{center}
  	\caption{Envy-free solutions might not exist under discontinuous utility functions }
  	\label{fig:disc}
  \end{figure}

\section{Example from \cref{figure:non-convex}}
\label{section: Non-convexity}
This section details an example which highlights that the set of prices that induce $\EF$-solutions can be nonconvex, though this set consists of a chain of polytopes which successively intersect; see \cref{figure:non-convex}.

The rent division instance consists of three agents with the following linear utilities for the three rooms; here, the rows of the $3 \times 3$ matrix given below correspond to the agents and the columns correspond to the rooms. 
\begin{align*}
\begin{bmatrix}
8-8x & 2-1.5x & 1-x  \\
1-x & 8-8x & 2-1.5x\\
2-1.5x & 1-x & 8-8x
\end{bmatrix}
\end{align*}
For this instance, the envy-free polytopes associated with the allocations  $\sigma_1 = (3,1,2), \sigma_2 =(2,3,1)$ and $\sigma_3 = (1,2,3)$ are nonempty. Note that \cref{algorithm:envy-free}, when executed on this instance, will consider these three permutations in order. In addition, as stated in  \cref{lemma:wts-drop}, we have $w({\sigma_1}) < w({\sigma_2}) < w({\sigma_3})$.

\end{document}